\documentclass[fullpaper,9pt]{IEEEtran}

\usepackage{amsmath,amssymb,enumerate,subfigure,multirow,hhline,color}
\usepackage{xcolor}
\usepackage{hyperref}
\usepackage{graphicx}
\usepackage{graphics}
\usepackage[sort]{cite}
\usepackage{amsthm}
\usepackage{algorithm}
\usepackage{algpseudocode}

\usepackage{tcolorbox}
\usepackage{epsfig,psfrag}
\usepackage{tabularx}
\usepackage{tabulary}

\usepackage[sort]{cite}

\usepackage{hyperref}
\hypersetup{breaklinks=true,colorlinks=true,linkcolor=blue,citecolor=blue}

\usepackage[noabbrev,capitalise,nameinlink]{cleveref}

\newtheorem{theorem}{Theorem}
\newtheorem{example}{Example}

\newtheorem{lemma}{Lemma}

\allowdisplaybreaks

\title{On the Local Quadratic Stability of T-S Fuzzy Systems in the Vicinity of the Origin}

\author{Donghwan Lee and Do Wan Kim% <-this % stops a space
\thanks{D. Lee is with the Department of Electrical Engineering,
KAIST, Daejeon, 34141, South Korea {\tt\small
donghwan@kaist.ac.kr}.}
\thanks{D. W. Kim is with the Department of Electrical Engineering, Hanbat National University, Daejeon 34158, South Korea {\tt\small
dowankim@hanbat.ac.kr}.}
}

\begin{document}

\maketitle

%=============================================================================================================
\begin{abstract}
The main goal of this paper is to introduce new local stability conditions for continuous-time Takagi-Sugeno (T-S) fuzzy systems. These stability conditions are based on linear matrix inequalities (LMIs) in combination with quadratic Lyapunov functions. Moreover, they integrate information on the membership functions at the origin and effectively leverage the linear structure of the underlying nonlinear system in the vicinity of the origin. As a result, the proposed conditions are proved to be less conservative compared to existing methods using fuzzy Lyapunov functions. Moreover, we establish that the proposed methods offer necessary and sufficient conditions for the local exponential stability of T-S fuzzy systems. The paper also includes discussions on the inherent limitations associated with fuzzy Lyapunov approaches. To demonstrate the theoretical results, we provide comprehensive examples that elucidate the core concepts and validate the efficacy of the proposed conditions.
\end{abstract}

\section{Introduction}

Over the past several decades, Takagi--Sugeno (T--S) fuzzy systems have gained prominence as effective tools for modeling nonlinear systems~\cite{khalil2002nonlinear}. These systems are characterized by a convex combination of linear subsystems, each weighted by nonlinear membership functions (MFs)~\cite{tanaka2004fuzzy}. Particularly within the framework of Lyapunov theory, T--S fuzzy systems offer a systematic approach for leveraging linear system theories in the stability analysis and control design of nonlinear systems. Moreover, these systems facilitate the formulation of problems based on convex linear matrix inequality (LMI) optimization procedures, which are readily solvable through standard convex optimization techniques~\cite{boyd1994linear}.
Due to these advantages, there has been a surge of interest in the stability analysis and control design of T--S fuzzy systems in recent years~\cite{nguyen2019fuzzy}.

Despite these advantages, T--S fuzzy model-based methods are often subject to a considerable degree of conservatism, which can be attributed to multiple sources: 1) The discrepancy between the infinite-dimensional parameterized matrix inequalities, which naturally arise in Lyapunov inequalities, and the numerically tractable finite-dimensional LMI problems introduces significant conservatism; 2) The constrained architectures of the Lyapunov functions and control laws can also be a source of conservatism; 3) Attempts to ensure stability across an overly expansive region of the state-space can lead to conservatism, particularly when the underlying nonlinear system only exhibits stability within a small region around the origin.

To reduce the conservatism, extensive research efforts have been undertaken. These efforts can be broadly classified into three primary categories:
\begin{enumerate}
\item \textbf{Relaxation of Parameterized LMIs:} Techniques for relaxing parameterized LMIs have also been developed. These include slack variable methods~\cite{fang2006new,kim2000new}, multidimensional convex summation techniques based on Polya's theorem~\cite{sala2007asymptotically,oliveira2007parameter}, as well as methods that leverage the structural information of membership functions~\cite{sala2007relaxed,narimani2009relaxed,kruszewski2009triangulation} to name just a few.

\item \textbf{Generalization of Lyapunov Functions and Control Laws:} Various extensions to traditional Lyapunov functions have been proposed, including piecewise Lyapunov functions~\cite{johansson1999piecewise,feng2003controller,campos2012new}, fuzzy Lyapunov functions (FLFs)~\cite{tanaka2003multiple,mozelli2009reducing,guerra2004lmi,ding2006further,lee2010improvement,lee2011new,lee2014generalized,lee2013relaxed,lee2010improvement,ding2010homogeneous}, line-integral Lyapunov functions~\cite{rhee2006new}, polynomial Lyapunov functions~\cite{Tanaka2009,sala2009polynomial,lam2011polynomial,bernal2011stability}, and augmented FLFs~\cite{kruszewski2008nonquadratic,lee2011approaches} to name just a few.

\item \textbf{Local Stability Approaches:} Some studies focus on estimating the domain of attraction through local stability analysis~\cite{bernal2010generalized,bernal2011stability,pan2011nonquadratic,lee2012fuzzy,lee2013relaxed}. By assessing stability within a localized region, these approaches aim to reduce conservatism by balancing the trade-off between the volume of the domain of attraction and the degree of conservatism.
\end{enumerate}

The primary aim of this paper is to introduce new local stability conditions for continuous-time T-S fuzzy systems by building upon the common quadratic Lyapunov function (QLF) approach~\cite{tanaka1997design}. While the QLF method is straightforward and widely used, it is known to introduce considerable conservatism. This conservatism is primarily due to the common Lyapunov matrix that should be found across all subsystems of the underlying fuzzy system.

In addressing these challenges, this paper demonstrates that the limitations associated with the common QLF approach can be mitigated. This is achieved by incorporating information about the membership functions at the origin into the quadratic stability conditions and by effectively leveraging the linear structure of the underlying nonlinear system in proximity to the origin. Consequently, the proposed quadratic stability conditions are theoretically proven to be less conservative than existing FLF methods in the literature.

The main contributions of this paper can be summarized as follows: 1) We develop several LMI stability conditions based on the QLF. Both theoretical and numerical evidences are provided to demonstrate that these quadratic stability conditions are less conservative compared to those based on FLFs in existing literature; 2) We rigorously establish that the proposed methods yield necessary and sufficient conditions for ensuring the local exponential stability of T-S fuzzy systems; 3) We identify and discuss the inherent limitations associated with both the QLF and FLF approaches. To substantiate these theoretical insights, comprehensive examples are provided to clarify the core concepts underpinning our contributions; 4) We demonstrate that the proposed QLF approach can be synergistically integrated with FLF methods to produce compounded effects, thereby reducing conservatism of existing stability conditions.

The organization of this paper is as follows: \cref{sec:preliminaries} introduces preliminary definitions, reviews existing quadratic and FLF approaches, and provides initial discussions to set the stage for the paper.
\cref{sec:main-results} contains the main results, detailing the proposed stability conditions along with associated examples and discussions to validate the theoretical contributions. \cref{sec:analysis} offers a comprehensive analysis, including discussions on the reduced conservatism of the proposed conditions, converse theorems, and the fundamental limitations of various approaches. This section also includes relevant examples to support the analysis. \cref{sec:combined-approach} integrates the two primary methods discussed in this paper—the FLF and QLF approaches—and analyzes their combined effects. Finally,~\cref{sec:conclusion} concludes the paper, offering remarks on potential topics for future research.

\section{Preliminaries}\label{sec:preliminaries}

\subsection{Notation}

The adopted notation is as follows. ${\mathbb R}$: sets of real
numbers; $A^T$: transpose of matrix $A$; $A \succ 0$ ($A \prec 0$,
$A \succeq 0$, and $A \preceq 0$, respectively): symmetric
positive definite (negative definite, positive semi-definite, and
negative semi-definite, respectively) matrix $A$; $I_n$ and $I$: $n\times n$ identity matrix and identity matrix of appropriate dimensions, respectively; $0$: zero vector of appropriate dimensions; $\lambda _{\min} (A)$ and $\lambda _{\max } (A)$:
minimum and maximum eigenvalues of symmetric matrix $A$, respectively.

\subsection{Takagi--Sugeno (T--S) fuzzy system}

Consider the continuous-time nonlinear system
\begin{align}
&\dot x(t) = f(x(t)),\label{nonlinear-system}
\end{align}
where $t \geq 0$ is the time, $x(t): = \begin{bmatrix} x_1(t) & \cdots  & x_n(t)\\
\end{bmatrix}^T \in {\mathbb R}^n$ is the state, $f:{\mathbb R}^n \to {\mathbb R}^n$ is a nonlinear function
such that $f(0) = 0$, i.e., the origin is an equilibrium
point of~\eqref{nonlinear-system}. By the sector nonlinearity
approach~\cite{tanaka2004fuzzy}, one can obtain the following T--S
fuzzy system representation of the nonlinear
system~\eqref{nonlinear-system} in the subset ${\cal X}$ of the
state-space ${\mathbb R}^n$:
\begin{align}
\dot x(t) = \sum_{i = 1}^r {{\alpha _i}(x(t)){A_i}x(t)} =: A(\alpha(t)) x(t) ,\quad \forall x(t) \in {\cal X}\label{fuzzy-system}
\end{align}
where $A_i \in {\mathbb R}^{n \times
n}, i\in \{1,2,\ldots, r \}$ are subsystem matrices, $i \in {\cal I}_r :=
\{ 1,\,2, \ldots ,\,r\}$ is the rule number, $\alpha_i :\,{\mathbb R}^n \to [0,\,1]$ is the membership function (MF) for each rule, and the vector of the MFs
\begin{align*}
\alpha (x): = \left[ {\begin{array}{*{20}{c}}
{{\alpha _1}(x)}\\
{{\alpha _2}(x)}\\
 \vdots \\
{{\alpha _r}(x)}
\end{array}} \right] \in {\mathbb R}^r
\end{align*}
lies in the unit simplex
\[
{\Lambda _r}: = \left\{ {\begin{array}{*{20}{c}}
{\alpha  \in {\mathbb R}^r: \sum\limits_{i = 1}^r {{\alpha _i}}  = 1,0 \le {\alpha _i} \le 1,i \in {\cal I}_r}
\end{array}} \right\}.
\]

In addition, ${\cal X} \in {\mathbb R}^n $ is a subset of the state-space ${\mathbb R}^n$ including the origin $0 \in {\cal X}$ and satisfies
\begin{align}
{\cal X} \subseteq \left\{ {x \in {\mathbb R}^n:f(x) = \sum_{i = 1}^r {{\alpha _i}(x)A_ix} } \right\}\label{region-X}
\end{align}
In other words, ${\cal X}$ is a modelling region where the T--S fuzzy model is defined. In this paper, we assume that ${\cal X}$ is compact.

\subsection{Quadratic stability}

Let us consider the quadratic Lyapunov function (QLF) candidate $V(x) = x^T P x$, where $P = P^T \succ 0$ is some matrix. Its time-derivative along the solution of~\eqref{nonlinear-system} is given as
\begin{align*}
\dot V(x(t)) = x{(t)^T}\sum_{i = 1}^r {{\alpha _i}(x(t))(A_i^TP + P{A_i})x(t)}
\end{align*}

Using Lyapunov direct method~\cite{khalil2002nonlinear}, the corresponding stability condition is given below.
\begin{lemma}[{\cite[Theorem~5]{tanaka2004fuzzy}}]
Suppose that there exists a symmetric matrix $P = P^T \in {\mathbb R}^{n\times n}$ such that the following LMIs hold:
\begin{align*}
&P \succ 0,\quad A_i^TP + P{A_i} \prec 0,\quad i \in {\cal I}_r.
\end{align*}
Then, the fuzzy system~\eqref{fuzzy-system} is locally asymptotically stable.
\end{lemma}

When the LMIs are feasible, then the system is locally asymptotically stable, and the domain of attraction (DA)~\cite{khalil2002nonlinear} is given by any Lyapunov sublevel sets defined by
\begin{align*}
{L_V}(c): = \{ x \in {\mathbb R}^n:V(x) \le c\}  \subseteq {\cal X}
\end{align*}
with any $c>0$ such that ${L_V}(c) \subseteq {\cal X}$. The largest DA is ${L_V}(c^*)$ with ${c^*}: = \max \{ c > 0:{L_V}(c) \subseteq {\cal X} \}$. While the LMI condition offers simplicity, a key drawback of the quadratic stability approach is its inherent conservatism. This conservatism primarily stems from the requirement for a common Lyapunov matrix $P$, which should satisfy the Lyapunov inequality across all subsystems of the fuzzy system. As an alternative, the fuzzy Lyapunov function (FLF) can be employed.

\subsection{Fuzzy Lyapunov stability}

Let us consider the Lyapunov function candidate
\begin{align*}
V(x) = {x^T}\left( {\sum_{i = 1}^r {{\alpha _i}(x)P_i} } \right)x =: {x^T}P(\alpha(x))x,
\end{align*}
which is called a fuzzy Lyapunov function (FLF), where $P_i = P_i^T \succ 0,i\in {\cal I}_r$.
Taking its time-derivative along the trajectory of~\eqref{fuzzy-system} leads to
\begin{align*}
\dot V(x(t)) =& x{(t)^T}\left[ {\sum\limits_{i = 1}^r {\sum\limits_{j = 1}^r {{\alpha _i}(x(t)){\alpha _j}(x(t))(A_i^T{P_j} + {P_j}{A_i})} } } \right.\\
&\left. {\sum\limits_{k = 1}^r {{{\dot \alpha }_k}(x(t)){P_k}} } \right]x(t).
\end{align*}

It is important to note that the use of the FLF presupposes that the MFs are differentiable, an assumption that does not universally hold. Therefore, when employing the FLF, we implicitly assume the differentiability of the MFs. A challenge that emerges in this context is that the stability condition depends on the derivatives of the MFs, which are not readily amenable to convexification. To circumvent this issue, \cite{tanaka2003multiple} proposed incorporating a bound on the derivatives of the MFs, $|{\nabla _x}{\alpha _i}(x)A(\alpha (x))x| \le {\varphi _i},i \in {\cal I}_r,x \in {\cal X}$. Building on this bound, a subsequent LMI-based stability condition has been introduced.
\begin{lemma}[{\cite[Theorem~1]{tanaka2003multiple}}]\label{lemma:Tanaka2003}
Suppose that the MFs are differentiable.
Moreover, suppose that there exist constants $\varphi_k > 0, k\in {\cal I}_r$, symmetric matrices $P_i = P_i^T \in {\mathbb R}^{n\times n}, i\in {\cal I}_r$ such that the following LMIs hold:
\begin{align*}
&P_i \succ 0,\quad i\in {\cal I}_r\\
&\sum\limits_{k = 1}^r {{\varphi _k}{P_k}}  + \frac{1}{2}\left( {A_i^T{P_j} + {P_j}{A_i} + A_j^T{P_i} + {P_i}{A_j}} \right)\prec 0,\\
&\forall (i,j) \in {\cal I}_r \times {\cal I}_r,\quad i \le j.
\end{align*}
Then, the fuzzy system~\eqref{fuzzy-system} is locally asymptotically stable.
\end{lemma}
In the Fuzzy Lyapunov Function (FLF) approach, the imposed bounds $|{\nabla _x}{\alpha _i}{(x)^T}A(\alpha (x))x| \le {\varphi _i},i \in {\cal I}_r$ imply that the matrix $A(\alpha(x))$ varies at a sufficiently slow rate. This constraint effectively limits the rate of change of the system's dynamics as represented by $A(\alpha(x))$, thereby facilitating the application of the FLF method for stability analysis. To proceed, let us define the set
\begin{align*}
\Omega (\varphi ): = \{ x \in {\cal X}:|{\nabla _x}{\alpha _i}{(x)^T}A(\alpha (x))x| \le {\varphi_i},i \in {\cal I}_r\}
\end{align*}
where $\varphi : = {\left[ {\begin{array}{*{20}{c}}
{{\varphi _1}}&{{\varphi _2}}& \cdots &{{\varphi _r}}
\end{array}} \right]^T}$. Clearly, this set is a subset of $\cal X$ such that the derivatives of the MFs are bounded within certain intervals. If the LMIs in~\cref{lemma:Tanaka2003} are feasible, then we have
\begin{align*}
&V(x)> 0,\quad \forall x \in {\cal X} \backslash \{ 0\},\\
&{\nabla _x}V{(x)^T}A(\alpha (x))x < 0,\quad \forall x \in \Omega (\varphi )\backslash \{ 0\}.
\end{align*}
Therefore, from the Lyapunov direct method~\cite{khalil2002nonlinear}, any Lyapunov sublevelset
\begin{align*}
{L_V}(c): = \{ x \in {\mathbb R}^n:V(x) \le c\},
\end{align*}
with some $c>0$ such that ${L_V}(c) \subseteq  \Omega (\varphi )$ is a domain of attraction (DA).
A fundamental question that naturally arises in this context is whether the set $\Omega(\varphi)$ includes a neighborhood of the origin, i.e., a ball centered at the origin. The significance of this question is directly tied to the practical utility of the conditions elaborated upon in this paper; if $\Omega(\varphi)$ does not include a neighborhood of the origin, then the proposed conditions would be useless for stability analysis around that point. Fortunately, the answer to this question is positive under certain mild assumptions.
\begin{theorem}\label{prop:1}
Suppose that the MFs $\alpha(x)$ are continuously differentiable on $\cal X$. Then, 1) the derivatives of the MFs are bounded in $\cal X$, 2) we have $\lim_{x \to 0_n} |{\nabla _x}{\alpha _i}(x)A(\alpha (x))x| = 0,\forall i \in {\cal I}_r$, and 3) there exists a ball ${\cal B}(\varepsilon): = \{ x \in {\mathbb R}^n:{\left\| x \right\|_2} \le \varepsilon \}$ of radius $\varepsilon>0$ centered at the origin such that ${\cal B}(\varepsilon) \subseteq \Omega(\varphi )$.
\end{theorem}

The proof is given in Appendix~\ref{appendix:2}.
\cref{prop:1} indicates that for the set $\Omega(\varphi )$ to include a neighborhood of the origin, it is not sufficient for the MFs to merely be differentiable; their derivatives must also exhibit continuity. To elucidate this point, we present an example featuring MFs that are differentiable over the set $\cal X$, but whose derivatives lack continuity and show its impact on the stability conditions.
\begin{example}\label{ex:5}
Let us consider the MFs
\begin{align*}
&{\alpha _1}(x) = \left\{ {\begin{array}{*{20}{c}}
{0,\quad {\rm{if}}{\mkern 1mu} {\mkern 1mu} x = 0}\\
{{x^2}\left( {\frac{1}{2}\sin \left( {\frac{\pi }{{2x}}} \right) + \frac{1}{2}} \right),\quad {\rm{if}}{\mkern 1mu} {\mkern 1mu} x \in [ - 1,1]\backslash \{ 0\}}
\end{array}} \right.\\
&{\alpha _2}(x) = 1 - {\alpha _1}(x)
\end{align*}
with ${\cal X} = [-1,1]$. The functions ${\alpha _1}$ and ${\alpha _2}$ are continuous on $\cal X$.
Its derivative in $x\in [-1,1]$ is given by
\begin{align*}
\frac{d}{{dx}}{\alpha _1}(x) = \left\{ {\begin{array}{*{20}{c}}
{0,\quad {\rm{if}}\quad x = 0}\\
\begin{array}{l}
2x\left( {\frac{1}{2}\sin \left( {\frac{\pi }{{2x}}} \right) + \frac{1}{2}} \right) - \frac{\pi }{2}\cos \left( {\frac{\pi }{{2x}}} \right),\\
{\rm{if}}\quad x \in [ - 1,1]\backslash \{ 0\}
\end{array}
\end{array}} \right.
\end{align*}
which is discontinuous because it oscillates as $x \to 0$. Therefore, $\lim_{x \to 0} |{\nabla _x}{\alpha _i}(x)A(\alpha (x))x|,\forall i \in {\cal I}_r$ do not exist. As a consequence, although $\Omega(\varphi )$ contains the origin, it does not contain any ball centered at the origin.
\end{example}

In~\cite{mozelli2009reducing}, an improved (less conservative) LMI condition has been proposed using the null effect
\begin{align}
\sum\limits_{k = 1}^r {{{\dot \alpha }_k}(x(t))}  = 0.\label{eq:slack1}
\end{align}
Using this property, one can add the slack term $\sum_{k = 1}^r {{{\dot \alpha }_k}(x(t))}M   = 0$ with any slack matrix variable $M$. For convenience and completeness of the presentation, the condition is formally introduced below.
\begin{lemma}[{\cite[Theorem~6]{mozelli2009reducing}}]\label{lemma:Mozelli2009}
Suppose that the MFs are continuously differentiable.
Moreover, suppose that there exist constants $\varphi_k > 0,k\in {\cal I}_r$, symmetric matrices $P_i = P_i^T \in {\mathbb R}^{n\times n}, i\in {\cal I}_r$ and $M = M^T \in {\mathbb R}^{n\times n}$ such that the following LMIs hold:
\begin{align}
&P_i \succ 0,\quad P_i + M \succeq 0,\quad i\in {\cal I}_r\nonumber\\
&\sum\limits_{k = 1}^r {{\varphi _k}{(P_k + M)}}  + \frac{1}{2}\left( {A_i^T{P_j} + {P_j}{A_i} + A_j^T{P_i} + {P_i}{A_j}} \right) \prec 0,\nonumber\\
&\forall (i,j) \in {\cal I}_r \times {\cal I}_r,\quad i \le j\label{eq:4}
\end{align}
Then, the fuzzy system~\eqref{fuzzy-system} is locally asymptotically stable.
\end{lemma}

As before, if the LMI condition in~\cref{lemma:Mozelli2009} is feasible, then any Lyapunov sublevelset ${L_V}(c)$ with some $c>0$ such that ${L_V}(c) \subseteq \Omega(\varphi)$ is a DA. As $\varphi_i$ approaches zero for $i\in {\cal I}_r$, the associated LMI condition in~\cref{lemma:Mozelli2009} become less conservative. However, this reduction in conservatism comes at a cost: the set $\Omega(\varphi )$ contracts towards the origin, leading to a corresponding reduction in the DA. This situation shows a trade-off between the level of conservatism in the stability conditions and the volume of the DA. In essence, achieving less conservative conditions results in a more restricted DA. A natural question arising here is if the conservatism can be completely eliminated as $\varphi_i$ approaches zero for $i\in {\cal I}_r$. The answer to this question is, in fact, negative, as demonstrated by the subsequent result.
\begin{theorem}[Fundamental limitation]\label{thm:conservatism}
Consider a nonlinear function~\eqref{nonlinear-system} that is locally exponentially stable around the origin and the corresponding fuzzy system~\eqref{fuzzy-system}. Suppose that its MFs are continuously differentiable on $\cal X$.
Then, corresponding the LMI condition of~\cref{lemma:Mozelli2009} may be infeasible for any $\varphi_i,i\in {\cal I}_r$.
\end{theorem}
The proof is given in~Appendix~\ref{appendix:9}.
\cref{thm:conservatism} establishes that the FLF approach contains an intrinsic conservatism that cannot be fully mitigated. Before concluding this section, we note two primary limitations of the FLF approach: First, its applicability is restricted to cases where the MFs are continuously differentiable. Second, as established in~\cref{thm:conservatism}, the approach exhibits inherent conservatism that cannot be fully eliminated. These limitations serve to delineate the scope and efficacy of the FLF approach in stability analysis. In contrast, the local QLF approaches proposed in the subsequent sections can overcome these disadvantages inherent to the FLF methods.

\section{Main results}\label{sec:main-results}

In this section, a new idea is proposed based on QLFs.
First of all, let us rewrite the system~\eqref{fuzzy-system} into
\begin{align*}
\dot x(t) =& A(\alpha (t))x(t)\\
=& \sum\limits_{i = 1}^r {{\alpha _i}(x(t)){A_i}x(t)} \\
=& \sum\limits_{i = 1}^r {({\alpha _i}(x(t)) - {\alpha _i}(0)){A_i}x(t)}  + \underbrace {\sum\limits_{i = 1}^r {{\alpha _i}(0){A_i}} }_{ = :{A_0}}x(t)\\
=& \sum\limits_{i = 1}^r {({\alpha _i}(x(t)) - {\alpha _i}(0)){A_i}x(t)}  + {A_0}x(t),
\end{align*}
which can be seen as a decomposition of the original system, $\dot x(t) = A(\alpha (x(t)))x(t)$, into the sum of the nominal linear system ${A_0}x(t)$ and the perturbed term $\sum_{i = 1}^r {({\alpha _i}(x(t)) - {\alpha _i}(0)){A_i}x(t)}$. The error term is expressed in terms of the difference between the MF ${\alpha _i}(x(t))$ at the current state $x(t)$ and the MF at the origin ${\alpha _i}(0)$.

Consequently, it is reasonable to conjecture that as $x(t)\to 0$, the error term diminishes, leaving only the nominal linear term. In this context, the nominal system can be considered as the limiting case of the fuzzy system described by~\eqref{fuzzy-system} as the state $x(t)$ approaches the origin $0$. If the MFs are continuous over the set $\cal X$, then as $x(t)\to 0$, the difference $|\alpha_i(x(t))- \alpha_i(0)| $ converges to zero for all $i\in {\cal I}_r$. Consequently, the error term $\sum_{i = 1}^r {({\alpha _i}(x(t)) - {\alpha _i}(0)){A_i}x(t)}$ also tends toward zero.
\begin{lemma}
Suppose that the MFs are continuous on $\cal X$. Then, $\lim_{x \to 0} |{\alpha _i}(x) - {\alpha _i}(0)| = 0,\forall i \in {\cal I}_r$.
\end{lemma}
This result is a direct consequence of the continuity of the MFs.
The following example shows the consequence when the continuity assumption is not met.
\begin{example}\label{eq:7}
Let us consider the MFs
\begin{align*}
{\alpha _1}(x) =& \left\{ {\begin{array}{*{20}{c}}
{\frac{1}{2},\quad {\rm{if}}{\mkern 1mu} {\mkern 1mu} \,\,x = 0}\\
{\frac{1}{2}\sin \left( {\frac{\pi }{{2x}}} \right) + \frac{1}{2},\quad {\rm{if}}{\mkern 1mu} {\mkern 1mu} \,\,x \in [-1,1]\backslash \{ 0\}}
\end{array}} \right.\\
{\alpha _2}(x) =& 1 - {\alpha _1}(x)
\end{align*}
with ${\cal X} = [-1,1]$. The functions ${\alpha _1}$ is discontinuous at $x = 0$.
The difference is given by
\begin{align*}
|{\alpha _1}(x) - {\alpha _1}(0)| = {\alpha _1}(x) = \frac{1}{2}\sin \left( {\frac{\pi }{{2x}}} \right),\quad x \in [-1,1]\backslash \{ 0\},
\end{align*}
which does not converge to $0$ as $x \to 0$.
Therefore, ${\cal H}(b)$ contains the origin, but does not contain any ball centered at the origin.
\end{example}

In this section, we consider the QLF candidate $V(x) = x^T P x$, where $P \succ 0$. Its time-derivative along the solution of~\eqref{fuzzy-system} is given as
\begin{align}
\dot V(x(t)) =& \sum\limits_{i = 1}^r {({\alpha _i}(x(t)) - {\alpha _i}(0))x{{(t)}^T}(A_i^TP + P{A_i})} x(t)\nonumber\\
&+ x{(t)^T}(A_0^TP + P{A_0})x(t)\label{eq:3}
\end{align}
By Lyapunov theorem~\cite{khalil2002nonlinear}, the asymptotic stability is guaranteed if $V(x(t)) >0,\forall x(t) \in {\cal X}\backslash \{ 0\} $ and $\dot V(x(t)) < 0,\forall x(t) \in {\cal X}\backslash \{ 0\}$, which is guaranteed if the following Lyapunov matrix inequalities hold:
\begin{align}
&P \succ 0,\nonumber\\
&\sum\limits_{i = 1}^r {({\alpha _i}(x)) - {\alpha _i}(0))(A_i^TP + P{A_i})}  + A_0^TP + P{A_0} \prec 0\label{eq:1}
\end{align}

In what follows, we will derive three sufficient LMI conditions based on~\eqref{eq:1}.
The proposed conditions can be categorized into two classes depending on how to treat the difference term $({\alpha _i}(x)) - {\alpha _i}(0))$, which are summarized as follows:
\begin{enumerate}
\item Polytopic bound (hyper rectangle bound):
\begin{align}
|{\alpha _i}(x)) - {\alpha _i}(0)| \le {b_i},\quad i \in {\cal I}_r\label{eq:2}
\end{align}
where $b_i\in [0,1], i\in {\cal I}_r$.
Equivalently, it can be written as
\begin{align*}
\left[ {\begin{array}{*{20}{c}}
{{\alpha _1}(x) - {\alpha _1}(0)}\\
 \vdots \\
{{\alpha _r}(x) - {\alpha _r}(0)}
\end{array}} \right]=:& \alpha (x) - \alpha (0)\in {\cal D}(b),
\end{align*}
where ${\cal D}(b)$ is a hyper rectangle defined as
\begin{align}
{\cal D}(b): = \{ \delta  \in {\mathbb R}^r:|{\delta _i}| \le b_i,i \in {\cal I}_r\}.\label{eq:rectangle}
\end{align}

\item Ball bound: for any $\eta>0$,
\begin{align*}
\sum_{i = 1}^r {{{({\alpha _i}(z) - {\alpha _i}(0))}^2}}  \le \eta.
\end{align*}
\end{enumerate}

Another important property we will use is the following null effect:
\begin{align}
\sum_{i = 1}^r {({\alpha _i}(x(t)) - {\alpha _i}(0))}  = 0\label{eq:null-property}
\end{align}
which is similar to the null effect of the derivatives of the MFs in~\eqref{eq:slack1}.
Such a geometric information can be used in the sequel in order to reduce conservatism of LMI stability conditions.

\subsection{Local quadratic stability based on the polytopic bound}

In this subsection, we introduce two LMI conditions. The first LMI condition employs all vertices of the polytopic bound as described in~\eqref{eq:2}. The second LMI condition utilizes an over-bounding technique, akin to the method outlined in~\cref{lemma:Mozelli2009}. We commence with the first setting. Specifically, by evaluating the condition at all vertices of the polytope, we can derive the following finite-dimensional LMI condition.
\begin{theorem}[Local quadratic stability~I]\label{thm:stability1}
Suppose that the MFs are continuous.
Moreover, suppose that there exist constants $b_i \in [0,1],i\in {\cal I}_r$, symmetric matrices $P = P^T \in {\mathbb R}^{n\times n}$ and $M=M^T \in {\mathbb R}^{n\times n}$ such that the following LMIs hold:
\begin{align}
&P\succ 0,\nonumber\\
&A_0^TP + P{A_0} + \sum\limits_{i = 1}^r {{\delta _i}(A_i^TP + P{A_i} + M)}  \prec 0,\nonumber\\
&({\delta _1},{\delta _2}, \ldots ,{\delta _r}) \in \{  - {b_1},{b_1}\}  \times \{  - {b_2},{b_2}\}  \times  \cdots  \times \{  - {b_r},{b_r}\}\label{eq:thm1:1}
\end{align}
Then, the fuzzy system~\eqref{fuzzy-system} is locally asymptotically stable.
\end{theorem}
The proof is given in Appendix~\ref{appendix:4}.
A notable difference between the condition in~\cref{thm:stability1} and those in fuzzy Lypaunov approaches is that the condition in~\cref{thm:stability1} uses information on the MFs, i.e., $\alpha (0)$, while the conditions derived from fuzzy Lypaunov approaches (e.g.,~\cref{lemma:Tanaka2003},~\cref{lemma:Mozelli2009}) do not depend on the MFs.
If the condition in~\cref{thm:stability1} is feasible, then it implies that there exists $P = P^T \succ 0$ such that
\begin{align*}
A{(\alpha (x))^T}P + PA(\alpha (x)) \prec 0,\quad \forall x \in {\cal H}(b),
\end{align*}
where
\begin{align*}
{\cal H}(b):=\{ x \in {\cal X}:|{\alpha _i}(x) - {\alpha _i}(0)| \le {b_i},i \in {\cal I}_r\},
\end{align*}
$b: = {\left[ {\begin{array}{*{20}{c}}
{{b_1}}&{{b_2}}& \cdots &{{b_r}}
\end{array}} \right]^T} \in {\mathbb R}^r$, and ${b_i} \in [0,1]$ for all $i\in {\cal I}_r$,
which corresponds to $\Omega(\varphi)$ in the FLF approach, while has several differences.
Moreover any Lyapunov sublevelset such that ${L_V}(c) \subseteq {\cal H}(b)$ with some $c>0$ is a DA.
Similar to the FLF methods, the proposed conditions are only useful when ${\cal H}(b)$ contains some neighborhood of the origin. In the sequel, we prove that ${\cal H}(b)$ includes neighborhoods of the origin under the continuity assumption.

Similar to the constraints associated with FLFs, the applicability of the proposed conditions hinges on whether the set ${\cal H}(b)$ includes a neighborhood of the origin. In what follows, we establish that ${\cal H}(b)$ does, in fact, encompass neighborhoods of the origin, provided that the continuity assumption holds. This validation serves to affirm the applicability and effectiveness of the proposed stability conditions.
\begin{theorem}\label{prop:2}
Suppose that the MFs are continuous on $\cal X$. Then, there exists a ball ${\cal B}(c)$ with some $c>0$ such that ${\cal B}(c) \subseteq {\cal H}(b)$ for any $b_i\in [0,1],i\in {\cal I}_r$.
\end{theorem}
\begin{proof}
From the definition of continuous functions, for any given $b_i \in [0,1],i \in {\cal I}_r$, we can find a ball ${\cal B}(c_i)$ with some $c_i >0$ such that $|{\alpha _i}(x) - {\alpha _i}(0)| \le {b_i},\forall x \in {\cal B}({c_i}),i \in {\cal I}_r$. Then, letting $c = {\min _{i \in {\cal I}_r }}{c_i}$ leads to ${\cal B}(c) \subseteq {\cal H}(b)$. This completes the proof.
\end{proof}

The following example shows the consequence when the continuity assumption is not met.
\begin{example}
Let us consider the MFs in~\cref{eq:7} again.
The function ${\alpha _1}$ is discontinuous at $x = 0$.
The difference $|{\alpha _1}(x) - {\alpha _1}(0)| = {\alpha _1}(x) = \frac{1}{2}\sin \left( {\frac{\pi }{{2x}}} \right), x \in [-1,1]\backslash \{ 0 \}$ does not converge to $0$ as $x \to 0$. Therefore, ${\cal H}(b)$ does not contain any ball centered at the origin.
\end{example}

Indeed, the continuity assumption is considerably less stringent than the requirement for continuous differentiability over the set $\cal X$ imposed by the FLF approach. This relaxed condition enhances the applicability and flexibility of the proposed stability analysis.
\begin{example}
Let us consider the MFs in~\cref{ex:5} again. The MFs are continuous in $\cal X$ and differentiable in $\cal X$. However, the derivatives are not continuous. As a consequence, $\Omega(\varphi )$ does not contain any ball centered at the origin. On the other hand, the difference is
\begin{align*}
&|{\alpha _1}(x) - {\alpha _1}(0)| = {\alpha _1}(x) = {x^2}\left( {\frac{1}{2}\sin \left( {\frac{\pi }{{2x}}} \right) + \frac{1}{2}} \right)
\end{align*}
for $x \in [-1,1]$, which vanishes as $x \to 0$. Therefore, one can prove that for any $b_i\in [0,1],i\in {\cal I}_r$, ${\cal H}(b)$ contains a ball centered at the origin.
\end{example}

\begin{example}\label{ex:2}
Let us consider the fuzzy system~\eqref{fuzzy-system} with
\begin{align*}
{A_1} = \left[ {\begin{array}{*{20}{c}}
{ - 5}&{ - 4}\\
{ - 1}&a
\end{array}} \right],\quad {A_2} = \left[ {\begin{array}{*{20}{c}}
{ - 4}&{ - 4}\\
{\frac{1}{5}(3b - 2)}&{\frac{1}{5}(3a - 4)}
\end{array}} \right],\\
A_3 = \left[ {\begin{array}{*{20}{c}}
{ - 3}&{ - 4}\\
{\frac{1}{5}(2b - 3)}&{\frac{1}{5}(2a - 6)}
\end{array}} \right],\quad {A_4} = \left[ {\begin{array}{*{20}{c}}
{ - 2}&{ - 4}\\
b&{ - 2}
\end{array}} \right],
\end{align*}
where ${\cal X} = \{ x \in {\mathbb R}^2:|{x_1}| \le \pi /2,|{x_2}| \le \pi /2\}$, and the MFs
\begin{align*}
{\alpha _1}(x) =& {\mu _1}(x){\mu _2}(x),\quad {\alpha _2}(x) = {\mu _1}(x){\beta _2}(x),\\
{\alpha _3}(x) =& {\beta _1}(x){\mu _2}(x),\quad {\alpha _4}(x) = {\beta _1}(x){\beta _2}(x),
\end{align*}
where ${\mu _i}(x) = (1 - \sin ({x_i}))/2, i \in \{ 1,2\}$ and ${\beta _i}(x) = 1 - {\alpha _i}(x), i \in \{ 1,2\}$.
In this case, ${\mu _i}(0) = 1/2, i \in \{ 1,2\}$, ${\beta _i}(0) = 1/2$, and
\begin{align*}
{\alpha _1}(0) =& 1/4,\quad {\alpha _2}(0) = 1/4,\\
{\alpha _3}(0) =& 1/4,\quad {\alpha _4}(0) = 1/4,
\end{align*}
The stability has been checked using~\cref{lemma:Mozelli2009} and~\cref{thm:stability1} for several values of pairs $(a,b)\in [-10,0]\times [0,200]$. For~\cref{lemma:Mozelli2009}, the derivative bounds are set to be $\varphi_i=0.85, i\in {\cal I}_r$. For~\cref{thm:stability1}, we consider the two settings, $b_i=0.2, i\in {\cal I}_r$ and $b_i=0.1, i\in {\cal I}_r$.

\cref{fig:1}(a) depicts a set of feasible points for several values of pairs $(a,b)\in [-10,0]\times [0,200]$ for~\cref{lemma:Mozelli2009} with $\varphi_i=0.85, i\in {\cal I}_r$ and for~\cref{thm:stability1} with $b_i=0.2, i\in {\cal I}_r$. It reveals that~\cref{lemma:Mozelli2009} and~\cref{thm:stability1} cover different regions, and one does not include the other region. Similarly, \cref{fig:1}(b) depicts the feasible region for~\cref{lemma:Mozelli2009} with $\varphi_i=0.85, i\in {\cal I}_r$ and for~\cref{thm:stability1} with $b_i=0.1, i\in {\cal I}_r$. In this case, the region of~\cref{thm:stability1} includes the region of~\cref{lemma:Mozelli2009}. We note that these results do not theoretically demonstrate that~\cref{thm:stability1} is less conservative than~\cref{lemma:Mozelli2009} because conservatism depends on the hyperparameters $\varphi_i, b_i, i\in {\cal I}_r$. Later, we will provide more rigorous and theoretical comparative analysis of conservatism of several approaches.

Next, let us fix $(a,b) = (-8,100)$, and compare the largest possible DAs estimated by~\cref{lemma:Mozelli2009} and~\cref{thm:stability1} in~\cref{fig:3}.
\cref{fig:3} shows an estimation of a DA using~\cref{lemma:Mozelli2009}. To obtain it, we computed the largest $\varphi^*  = {\varphi _i},i \in {\cal I}_r$ such that~\cref{lemma:Mozelli2009} is feasible, which is approximately $\varphi^* = 4.0789$. Next, we computed the largest $c>0$ such that ${L_V}(c) \subseteq \Omega (\varphi )$ with $\varphi  = {\left[ {\begin{array}{*{20}{c}}
{{\varphi ^*}}& \cdots &{{\varphi ^*}}
\end{array}} \right]^T}$. The boundary of such a sublevel set ${L_V}(c)$ is highlighted with the red solid line. Moreover, the region $\Omega (\varphi )$ is highlighted with the grey color. We note that the region $\Omega (\varphi )$ depends on the system matrices $A_i, i\in {\cal I}_r$. Therefore, $\Omega (\varphi )$ may change with across different pairs of $(a,b)$.

\cref{fig:3}(b) depicts an estimation of a DA using~\cref{thm:stability1}. To calculate it, we computed the largest $b^*  = {b_i} \in [0,1],i \in {\cal I}_r$ such that~\cref{thm:stability1} is feasible, which is approximately $b^* = 0.2603$. Next, we computed the largest $c>0$ such that ${L_V}(c) \subseteq {\cal H} (B)$ with $b = {\left[ {\begin{array}{*{20}{c}}
{{b^*}}& \cdots &{{b^*}} \end{array}} \right]^T}$. The boundary of such a sublevel set ${L_V}(c)$ is highlighted with the red solid line, and ${\cal H} (b)$ is highlighted with the grey color. Contrary to $\Omega (\varphi )$, ${\cal H} (b)$ does not depend on the system matrices. As can be seen from the figures, the DA estimation from~\cref{thm:stability1} is larger than that from~\cref{lemma:Mozelli2009}.
\begin{figure}
\centering\subfigure[$\varphi_i=0.85$,
$b_i=0.2,i\in {\cal I}_r$.]{\includegraphics[width=0.3\textwidth]{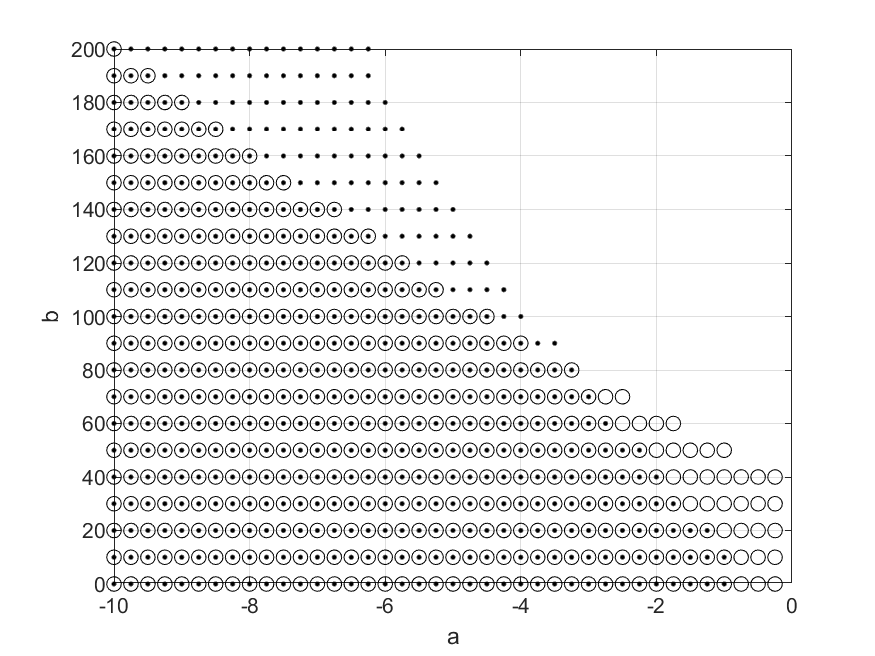}}
\centering\subfigure[$\varphi_i=0.85$, $b_i=0.1,i\in {\cal I}_r$.]{\includegraphics[width=0.3\textwidth]{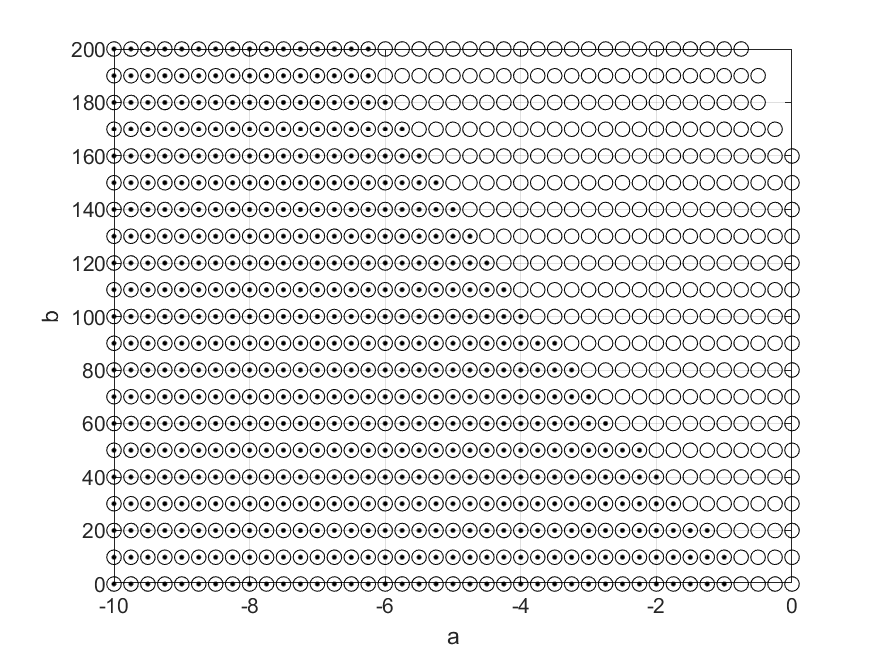}}
\caption{\cref{ex:2}: Feasible region for several values of $(a,b)\in [-10,0]\times [0,200]$.
`$\bullet$' indicates a feasible point using~\cref{lemma:Mozelli2009} with $\varphi_i=0.85, i\in {\cal I}_r$.
`$\circ$' indicates a feasible point using~\cref{thm:stability1}.}\label{fig:1}
\end{figure}

\begin{figure}
\centering\subfigure[The DA estimate from~\cref{lemma:Mozelli2009}. The grey region indicates $\Omega (\varphi )$ with $(a,b) =(-8,100)$ and $\varphi_i=4.0789,i\in {\cal I}_r$. The red line indicates the boundary of the Lyapunov sublevel set ${L_V}(c) \subseteq \Omega (\varphi )$.]{\includegraphics[width=0.3\textwidth]{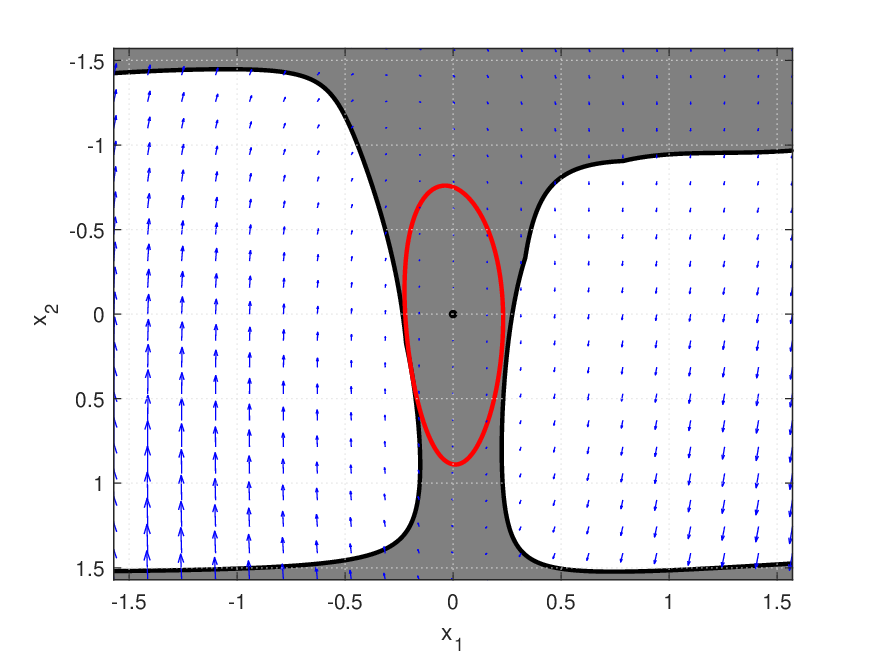}}
\centering\subfigure[The DA estimate from~\cref{thm:stability1}. The grey region indicates ${\cal H}(b)$ with $(a,b) =(-8,100)$ and $b_i=0.2603,i\in {\cal I}_r$. The red line indicates the boundary of the Lyapunov sublevel set ${L_V}(c) \subseteq {\cal H} (b)$.]{\includegraphics[width=0.3\textwidth]{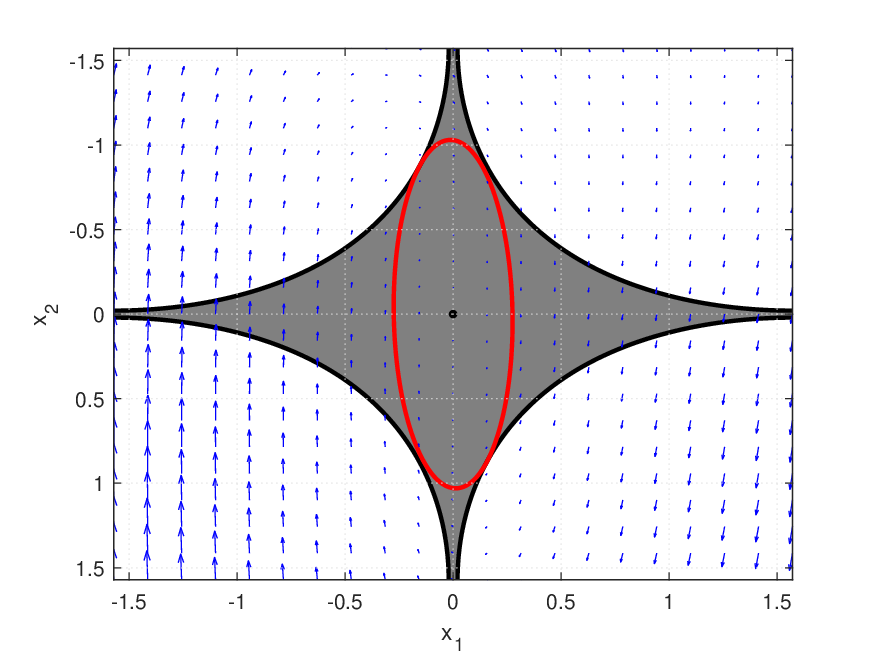}}
\caption{\cref{ex:2}: The DAs estimated by (a)~\cref{lemma:Mozelli2009} and (b)~\cref{thm:stability1} with $(a,b) = (-8,100)$.
The red line indicates the boundaries of the Lyapunov sublevel sets.
 }\label{fig:3}
\end{figure}

We note that through this single example, the conservativeness cannot be fairly compared, while it may give some insights.
In the remaining parts of the paper, we provide some theoretical analysis of conservatism of different LMI conditions.
\end{example}

The condition in~\cref{thm:stability1} necessitates solving the LMIs at all possible vertices of the hyper-rectangle described by~\eqref{eq:rectangle}. While this approach is thorough, this approach can be cumbersome to implement and computationally inefficient. A more efficient, albeit potentially more conservative, LMI condition can be achieved through the use of over-bounding techniques. These techniques obviate the need to check LMIs at every vertex of the hyper-rectangle in~\eqref{eq:rectangle}, thereby reducing the numerical and implementational complexities of~\cref{thm:stability1}.
\begin{theorem}[Local quadratic stability~II]\label{thm:stability2}
Suppose that the MFs are continuous.
Moreover, suppose that there exist constants $b_i \in [0,1],i\in {\cal I}_r$, symmetric matrices $P=P^T \in {\mathbb R}^{n\times n}$ and $M=M^T \in {\mathbb R}^{n\times n}$ such that the following LMIs hold:
\begin{align}
&P \succ 0,\nonumber\\
&A_i^TP + P{A_i} + M \prec 0,\quad i \in {\cal I}_r\label{eq:thm2:1}\\
&A_0^TP + P{A_0} - \sum\limits_{i = 1}^r {{b_i}(A_i^TP + P{A_i} + M)}  \prec 0\label{eq:thm2:2}
\end{align}
Then, the fuzzy system~\eqref{fuzzy-system} is locally asymptotically stable.
\end{theorem}
\begin{proof}
Suppose that the LMIs are feasible. Then, due to~\eqref{eq:thm2:1}, one can conclude that~\eqref{eq:thm2:2} holds for $b_i, i\in {\cal I}_r$ replaced with all possible combinations $\pm b_i, i\in {\cal I}_r$. Therefore, the LMI in~\eqref{eq:thm1:1} hold. The remaining parts of the proof follow the same lines as in the proof of~\cref{thm:stability1}. This completes the proof.
\end{proof}

\begin{example}\label{ex:3}
Let us consider the system in~\cref{ex:2} again.
The stability has been checked using~\cref{lemma:Mozelli2009} and~\cref{thm:stability2} for several values of pairs $(a,b)\in [-10,0]\times [0,200]$. We consider $\varphi_i=0.85, i\in {\cal I}_r$ for~\cref{lemma:Mozelli2009} and $b_i=0.1, i\in {\cal I}_r$ for~\cref{thm:stability2}. \cref{fig:4} depicts a set of feasible points for several values of pairs $(a,b)\in [-10,0]\times [0,200]$, which reveal that the region of~\cref{thm:stability1} includes the region of~\cref{lemma:Mozelli2009}.
\begin{figure}
\centering
\includegraphics[width=0.3\textwidth]{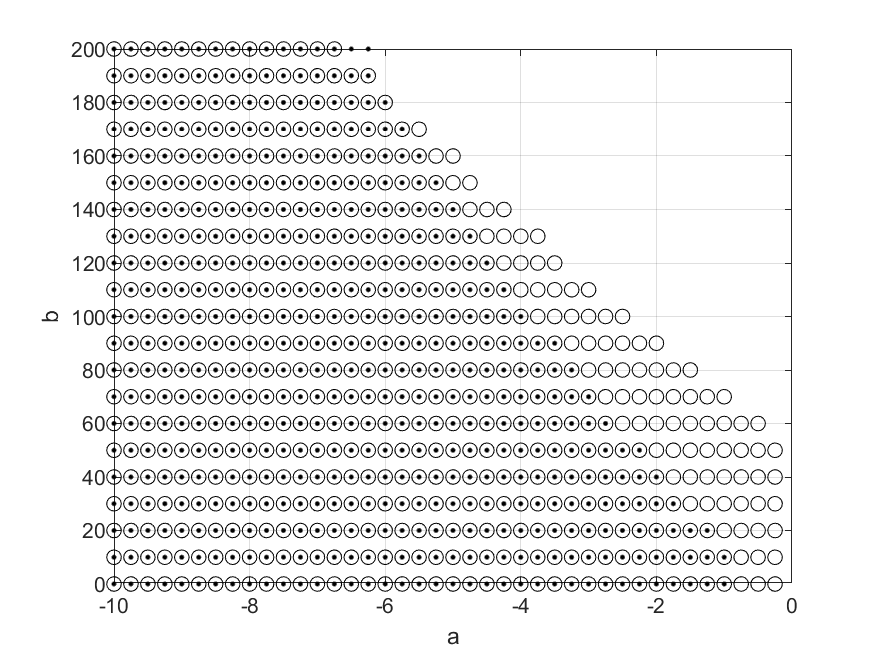}
\caption{\cref{ex:3}: Feasible region for several values of $(a,b)\in [-10,0]\times [0,200]$.
`$\bullet$' indicates the feasible point of~\cref{lemma:Mozelli2009} with $\varphi_i=0.85, i\in {\cal I}_r$.
`$\circ$' indicates the feasible point of~\cref{thm:stability2} with $b_i=0.1, i\in {\cal I}_r$.}\label{fig:4}
\end{figure}
\end{example}

\subsection{Local quadratic stability based on the ball bound}

Until now, we have explored two versions of LMI stability conditions, both of which are based on the polytopic bound ${\cal H}(b)$ for the difference $\alpha(x) - \alpha(0)$. We will now introduce a third version that is based on a ball bound for the difference $\alpha(x) - \alpha(0)$. This alternative approach provides another avenue for stability analysis, offering a different set of trade-offs in terms of conservatism and computational complexity.
\begin{theorem}[Local quadratic stability~III]\label{thm:stability3}
Suppose that the MFs are continuous.
Moreover, suppose that there exist a constant $\eta > 0$, symmetric matrices $P=P^T \in {\mathbb R}^{n\times n}$, $G=G^T\in {\mathbb R}^{n\times n}$, and a matrix $M\in {\mathbb R}^{n\times n}$ such that the following LMIs hold:
\begin{align}
&P \succ 0,\quad \left[ {\begin{array}{*{20}{c}}
{\eta G + P{A_0} + A_0^TP}&{{\Gamma ^T}}\\
\Gamma &{ - {I_r} \otimes G}
\end{array}} \right] \prec 0,\label{eq:8}
\end{align}
where ${\Gamma ^T}: = \left[ {\begin{array}{*{20}{c}}
{{{(P{A_1} + M)}^T}}& \cdots &{{{(P{A_r} + M)}^T}}
\end{array}} \right]$. Then, the fuzzy system~\eqref{fuzzy-system} is locally asymptotically stable.
\end{theorem}

The proof is given in Appendix~\ref{appendix:3}. If the LMI condition is feasible, then any Lyapunov sublevelset ${L_V}(c)$ with some $c>0$ such that ${L_V}(c) \subseteq {\cal U}(\eta)$ is a DA, where
\begin{align*}
{\cal U}(\eta ): = \left\{ {x \in {\mathbb R}^n:\sum_{i = 1}^r {{{({\alpha _i}(x) - {\alpha _i}(0))}^2}}  \le \eta } \right\}.
\end{align*}

\begin{theorem}
Suppose that the MFs are continuous on $\cal X$. Then, there exists a ball ${\cal B}(c)$ with some $c>0$ such that ${\cal B}(c) \subseteq {\cal U}(\eta)$ for any $\eta\in [0,\infty)$.
\end{theorem}
\begin{proof}
The proof can be completed following similar arguments as in the proof of~\cref{prop:2}.
\end{proof}

\begin{example}\label{ex:4}
Let us consider the system in~\cref{ex:2} again.
The stability has been checked using~\cref{lemma:Mozelli2009} and~\cref{thm:stability3} for several values of pairs $(a,b)\in [-10,0]\times [0,200]$. For~\cref{lemma:Mozelli2009}, the derivative bounds are set to be $\varphi_i=0.85, i\in {\cal I}_r$. For~\cref{thm:stability3}, we consider $\eta=0.9$.
\cref{fig:5} depicts a set of feasible points for several values of pairs $(a,b)\in [-10,0]\times [0,200]$.
It reveals that the region of~\cref{thm:stability3} includes the region of~\cref{lemma:Mozelli2009}.
\begin{figure}
\centering
\includegraphics[width=0.3\textwidth]{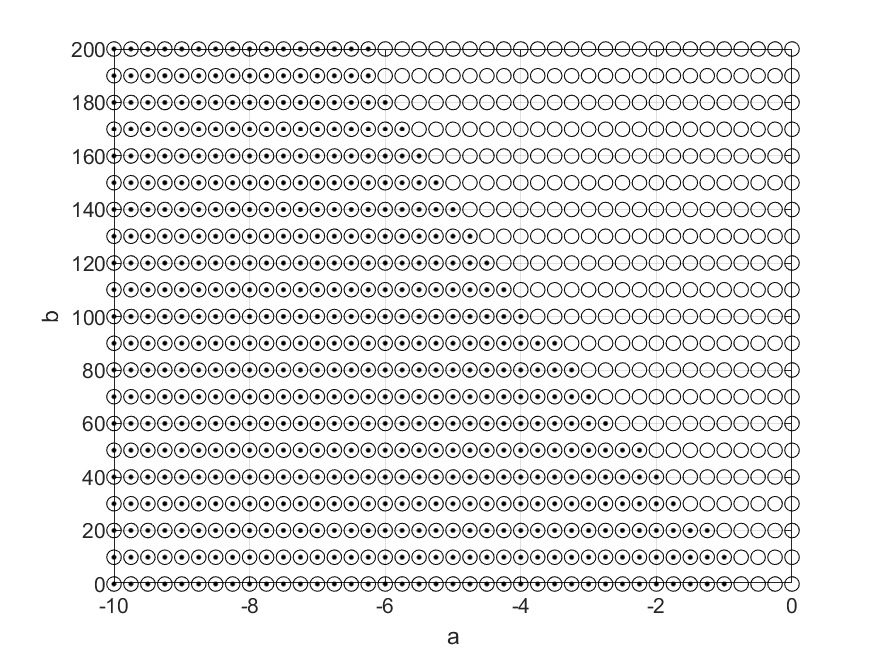}
\caption{\cref{ex:4}: Feasible region for several values of $(a,b)\in [-10,0]\times [0,200]$.
`$\cdot$' indicates~\cref{lemma:Mozelli2009} with $\varphi_i=0.85, i\in {\cal I}_r$.
`$\circ$' indicates~\cref{thm:stability3} with $\eta=0.9$.}\label{fig:5}
\end{figure}

Similar to~\cref{ex:2}, let us fix $(a,b) = (-8,100)$, and compare the largest possible DAs estimated by~\cref{lemma:Mozelli2009} with $\varphi_i=4.0789, i\in {\cal I}_r$ and~\cref{thm:stability3} with $\eta = 1.1645$. \cref{fig:6}(a) shows an estimate of a DA using~\cref{lemma:Mozelli2009}, and \cref{fig:6}(b) depicts an estimate of a DA using~\cref{thm:stability1}. As can be seen from the figures, the DA estimate from~\cref{thm:stability3} is larger than that from~\cref{lemma:Mozelli2009}. Indeed, the DA estimate using~\cref{thm:stability3} in~\cref{fig:6}(b) is the largest one in this paper.
\begin{figure}
\centering\subfigure[Region of the set $\Omega (\varphi ): = \{ x \in {\cal X}:|{\nabla _x}{\alpha _i}{(x)^T}A(\alpha (x))x| \le {\varphi _k},k \in {\cal I}_r\}$  from~\cref{lemma:Mozelli2009} with $(a,b) =(-10,0)$ and $\varphi_i=4.0789, i\in {\cal I}_r$. The red line indicates the boundary of the Lyapunov sublevel set ${L_V}(c) \subseteq \Omega(\varphi)$.]
{\includegraphics[width=0.3\textwidth]{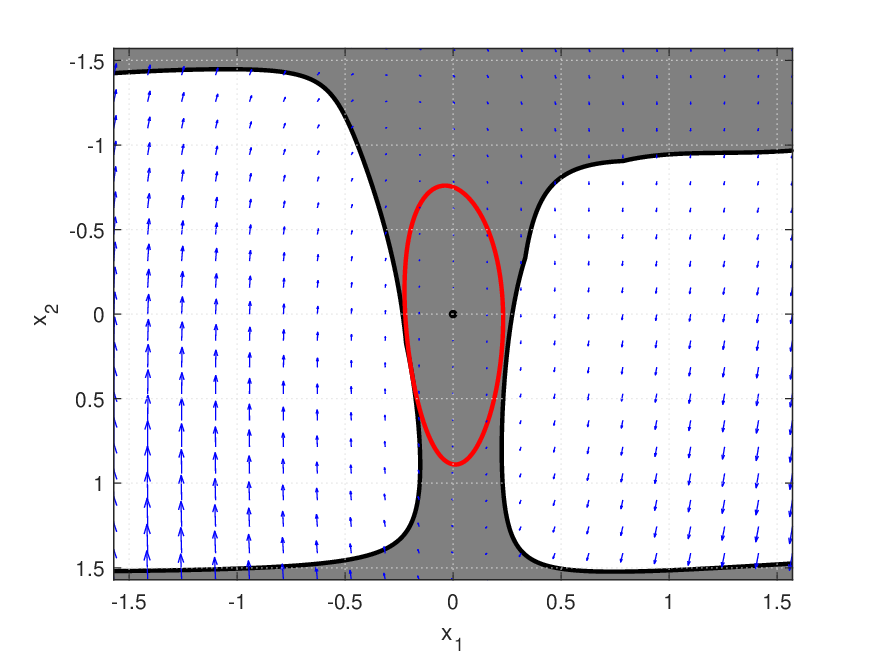}}
\centering\subfigure[Region of the set ${\cal H}(b): = \{ x \in {\cal X}:|{\alpha _i}(x) - {\alpha _i}(0)| \le {b_i},i \in {\cal I}_r\}$ from~\cref{thm:stability3} with $(a,b) =(-10,0)$ and $\eta = 1.1645$.
The red line indicates the boundary of the Lyapunov sublevel set ${L_V}(c) \subseteq {\cal U}(\eta)$.]
{\includegraphics[width=0.3\textwidth]{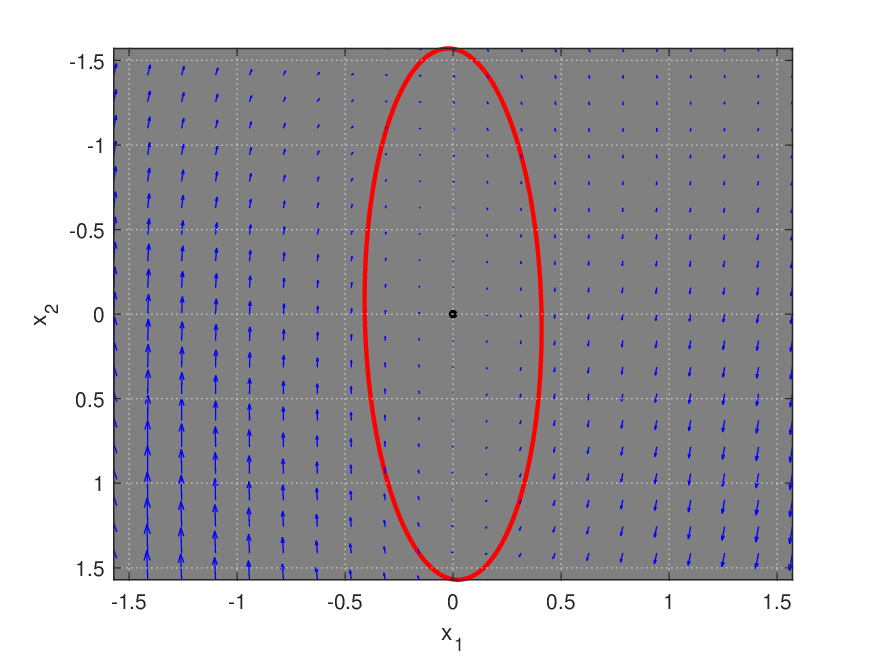}}
\caption{\cref{ex:4}: The DAs estimated by (a)~\cref{lemma:Mozelli2009} and (b)~\cref{thm:stability3} with $(a,b) = (-8,100)$.
The red line indicates the boundaries of the Lyapunov sublevel sets.}\label{fig:6}
\end{figure}

\end{example}

Based on the findings thus far, the differences and advantages of the proposed quadratic approaches in comparison to the FLF methods can be summarized as follows: 1) The proposed quadratic methods are applicable even when the MFs are not differentiable, thereby broadening the scope of systems to which these methods can be applied; 2) The bound $|{\alpha _i}(x) - {\alpha _i}(0)| \le {b_i}$ is solely dependent on the MFs, in contrast to $|{\nabla _x}{\alpha _i}{(x)^T}A(\alpha (x))x| \le {\varphi _i},i \in {\cal I}_r$, which is influenced by the system dynamics $A(\alpha (x))$. Consequently, the sets ${\cal H}(b)$ and ${\cal U}(\eta)$ are independent of the system matrices $A_i, i\in {\cal I}_r$, whereas $\Omega (\varphi )$ is dependent on both the system matrices and the MFs.
These distinctions underscore the flexibility and applicability of the proposed quadratic approaches, particularly in scenarios where the FLF methods may be limited.

\section{Analysis}\label{sec:analysis}
In the FLF approaches,~\cref{thm:conservatism} reveals that conditions based on FLFs may fail to identify local stability even when the system is, in fact, locally exponentially stable. This limitation highlights an inherent conservatism in the FLF methods.
Conversely, in the proposed quadratic approach, it can be proven that under certain mild assumptions, the converse statement holds true. Specifically, if the original nonlinear system is locally exponentially stable, then the LMI conditions outlined in~\cref{thm:stability1},~\cref{thm:stability2}, and~\cref{thm:stability3} are feasible. This feature enhances the applicability of the proposed quadratic stability methods.
\begin{theorem}[Converse theorem~I]\label{thm:converse1}
Suppose that the MFs are differentiable on $\cal X$, the origin is an equilibrium point of~\eqref{nonlinear-system}, and $f$ is continuously differentiable in some neighborhood of the origin. Then, the original nonlinear system~\eqref{nonlinear-system} is locally exponentially stable around the origin if and only if the LMI condition of~\cref{thm:stability1} is feasible for sufficiently small $b_i\in [0,1],i\in {\cal I}_r$.
\end{theorem}
The proof is given in Appendix~\ref{appendix:5}. By employing a similar line of reasoning, one can also establish the non-conservatism of the remaining two conditions,~\cref{thm:stability2} and~\cref{thm:stability3}.
\begin{theorem}[Converse theorem~II]\label{thm:converse2}
Suppose that the MFs are differentiable on $\cal X$, the origin is an equilibrium point of~\eqref{nonlinear-system}, and $f$ is continuously differentiable in some neighborhood of the origin.
\begin{enumerate}
\item Then, the original nonlinear system~\eqref{nonlinear-system} is locally exponentially stable around the origin if and only if the LMI condition of~\cref{thm:stability2} is feasible for sufficiently small $b_i \in [0,1],i\in {\cal I}_r$.

\item Moreover, the original nonlinear system~\eqref{nonlinear-system} is locally exponentially stable around the origin if and only if the LMI condition of~\cref{thm:stability3} is feasible for sufficiently small $\eta >0$.
\end{enumerate}
\end{theorem}
The proof is given in Appendix~\ref{appendix:7}.
In the following, we establish less conservatism of~\cref{thm:stability1} compared to~\cref{lemma:Mozelli2009}.
\begin{theorem}[Feasibility inclusion]\label{thm:non-conservatism2}
Suppose that the MFs are differentiable on $\cal X$, the origin is an equilibrium point of~\eqref{nonlinear-system}, and $f$ is continuously differentiable in some neighborhood of the origin.
Suppose that the LMI condition of~\cref{lemma:Mozelli2009} is feasible for some constants $\varphi_i > 0,i\in {\cal I}_r$.
Then, the LMI condition of~\cref{thm:stability1} is also feasible for sufficiently small $b_i\in [0,1],i\in {\cal I}_r$.
However, the conserve does not hold in general.
\end{theorem}

The proof is given in Appendix~\ref{appendix:6}. \cref{thm:non-conservatism2} theoretically proves that~\cref{thm:stability1} is less conservative than~\cref{lemma:Mozelli2009}. Similar arguments can be applied for~\cref{thm:stability2} and~\cref{thm:stability3}, but the results are omitted here for brevity.

\cref{thm:conservatism} establishes some fundamental limitations of the FLF approaches.
Despite the advantages of the proposed QLF approaches, both QLF and FLF share some fundamental limitations. Specifically, when employing either QLFs or FLFs, there exists a class of nonlinear systems whose stability cannot be identified through the associated LMI stability conditions.
\begin{theorem}[Fundamental limitations]\label{thm:fundamental-limit2}
Let us consider the nonlinear system~\eqref{nonlinear-system}, and suppose that the origin is a locally asymptotically but non-exponentially stable equilibrium point of~\eqref{nonlinear-system}. Consider the corresponding fuzzy system~\eqref{fuzzy-system}.
Then, both the FLF approaches (\cref{lemma:Tanaka2003} and \cref{lemma:Mozelli2009}) and the quadratic approaches (\cref{thm:stability1}, \cref{thm:stability2}, and \cref{thm:stability3}) fail to identify the stability, i.e., the corresponding LMI conditions are infeasible.
\end{theorem}
The proof is given in Appendix~\ref{appendix:8}.
In the following, we present simple examples to demonstrate the results in~\cref{thm:fundamental-limit2}.
\begin{example}\label{ex:6}
Let us consider the nonlinear system~\eqref{nonlinear-system} with $f(x)= -x^3$, whose origin is asymptotically but non-exponentially stable because its solution is $x(t) = \frac{{x(0)}}{{\sqrt {1 + 2x{{(0)}^2}t} }}$.
The corresponding fuzzy system is~\eqref{fuzzy-system} with ${\alpha _1}(x) = {x^2},{\alpha _1}(x) = 1 - {x^2},{A_1} =  - 1,{A_2} = 0,{\cal X} = [ - 1,1]$. For this system, all the conditions in this paper are infeasible with any possible hyperparameters.
For example, let us consider the condition of~\cref{thm:stability1}. In this case, $A_0 = 0$, and the LMIs in~\eqref{eq:thm1:1} are
${\delta _1}(A_1^TP + P{A_1} + M) + {\delta _2}(A_2^TP + P{A_2} + M) = {\delta _1}( - 2P + M) + {\delta _2}M < 0$, which together with $P>0$ leads to $({\delta _1} + {\delta _2})M < 2{\delta _1}P < 0$. Since the condition should hold for all $({\delta _1},{\delta _2}) \in \{  - {b_1},{b_1}\}  \times \{  - {b_2},{b_2}\}$, when $M\geq 0$, then the condition with $({\delta _1},{\delta _2}) = ({b_1},{b_2})$ cannot be satisfied. In addition, when $M< 0$, then the condition with $({\delta _1},{\delta _2}) = ({-b_1},{-b_2})$ cannot be satisfied.
\end{example}

We can extend the concepts in~\cref{ex:6} and~\cref{thm:fundamental-limit2}, and provide some ways to classify the set of fuzzy systems with $A_i,i\in {\cal I}_r$ whose stability cannot be identified through the LMI conditions based on the FLFs. The results are given in the following theorem, which can be seen as a second version of~\cref{thm:conservatism}.
\begin{theorem}[Fundamental limitations]\label{thm:fundamental-limit3}
Let us consider the fuzzy system~\eqref{fuzzy-system} and suppose the system matrices $A_i,i\in {\cal I}_r$ are given.
If there exist some MFs $\alpha_i(x),i\in {\cal I}_r$ defined in $\cal X$ such that the nonlinear system~\eqref{nonlinear-system} with $f(x) = \sum_{i = 1}^r {{\alpha _i}(x){A_i}} ,x \in {\cal X}$ is locally asymptotically but non-exponentially stable around the origin, then for the system matrices $A_i,i\in {\cal I}_r$, the FLF approaches (\cref{lemma:Tanaka2003} and \cref{lemma:Mozelli2009}) fail to identify the stability, i.e., the corresponding LMI conditions are infeasible.
\end{theorem}

The proof is a direct corollary of~\cref{thm:fundamental-limit2} and is therefore omitted for brevity.
The following example demonstrates~\cref{thm:fundamental-limit3}.
\begin{example}\label{ex:7}
Let us consider the linear system $\dot x(t) = -x(t)$, which is globally exponentially stable.
The corresponding fuzzy model~\eqref{fuzzy-system} is given with $A_1 = -1, A_2 = 0, \alpha_1(x) = 1, \alpha_2(x) = 0$.
For this fuzzy system, the FLF-based conditions (\cref{lemma:Tanaka2003} and \cref{lemma:Mozelli2009}) are infeasible.
This is because we can find MFs, ${\alpha _1}(x) = {x^2},{\alpha _1}(x) = 1 - {x^2}$ in~\cref{ex:6}, so that $\sum_{i = 1}^r {{\alpha _i}(x){A_i}}  = f(x) =  - {x^3},x \in {\cal X}$, and for the corresponding nonlinear system~\eqref{nonlinear-system} with $f(x)= -x^3$, whose origin is asymptotically stable but non-exponentially stable.

Let us consider the nonlinear system~\eqref{nonlinear-system} with $f(x) = \left[ {\begin{array}{*{20}{c}}
{{x_2}}\\
{ - {x_1} + \mu (1 + {x_2}){x_2}/2}
\end{array}} \right]$ and $\mu = -2$. By the Lyapunov's indirect method, since
\begin{align*}
{\left. {{\nabla _x}f(x)} \right|_{x = 0}} = \left[ {\begin{array}{*{20}{c}}
0&1\\
{ - 1}&{\mu /2}
\end{array}} \right] = A
\end{align*}
is Hurwitz, the system is locally exponentially stable at the origin.
The corresponding fuzzy system~\eqref{fuzzy-system} is defined with
\begin{align}
{A_1} = \left[ {\begin{array}{*{20}{c}}
0&1\\
{ - 1}&\mu
\end{array}} \right],\quad {A_2} = \left[ {\begin{array}{*{20}{c}}
0&1\\
{ - 1}&0
\end{array}} \right],\label{eq:14}
\end{align}
and ${\alpha _1}(x) = (1 + {x_2})/2,{\alpha _2}(x) = 1 - {\alpha _1}(x)$, ${\cal X} = \{ x \in {\mathbb R}^2:{x_1} \in [ - 1,1],{x_2} \in [ - 1,1]\} $.

Using~\cref{thm:fundamental-limit3}, we can prove that the FLF-based conditions (\cref{lemma:Tanaka2003}, \cref{lemma:Mozelli2009}, and other conditions in the literature) are infeasible.
To prove this, consider the MFs ${\alpha _1}(x) = 1 - x_1^2,{\alpha _2}(x) = 1 - {\alpha _1}(x)$ with ${\cal X} = \{ x \in {\mathbb R}^2:{x_1} \in [ - 1,1],{x_2} \in [ - 1,1]\}$. The corresponding nonlinear system~\eqref{nonlinear-system} with $f(x) = \left[ {\begin{array}{*{20}{c}}
{{x_2}}\\
{\mu (1 - x_1^2){x_2} - {x_1}}
\end{array}} \right]$ is the well-known van der Pol oscillator, which is known to be locally asymptotically but non-exponentially stable around the origin.
Therefore, by~\cref{thm:fundamental-limit3}, the FLF-based conditions should be infeasible.
\end{example}

A source of conservatism of the FLF-based conditions (e.g.,~\cref{lemma:Tanaka2003} and~\cref{lemma:Mozelli2009}) arises from the lack of information on the MFs. These conditions are solely dependent on the system matrices $A_i, i\in {\cal I}_r$, and do not incorporate any information about the structure of the MFs. In contrast, the proposed local QLF-based conditions take into account specific information about the structure of the MFs at the origin. As a result, the proposed QLF conditions can reduce the conservatism inherent in the FLF-based approaches.
\begin{theorem}[Fundamental limitations]\label{thm:fundamental-limit4}
Let us consider the fuzzy system~\eqref{fuzzy-system}, and suppose that
\begin{enumerate}
\item the system matrices $A_i,i\in {\cal I}_r$ are given;

\item the MFs are not given, but their values at the origin $\alpha_i(0)=\beta_i,i\in {\cal I}_r$ are given.
\end{enumerate}

If there exist some MFs $\alpha_i(x),i\in {\cal I}_r$ defined in $\cal X$ such that $\alpha_i(0)=\beta_i,i\in {\cal I}_r$ and the nonlinear system~\eqref{nonlinear-system} with $f(x) = \sum_{i = 1}^r {{\alpha _i}(x){A_i}} ,x \in {\cal X}$ is locally asymptotically but non-exponentially stable around the origin. Then, for the system matrices $A_i,i\in {\cal I}_r$ and MFs' values at the origin, the QLF approaches (\cref{thm:stability1}, \cref{thm:stability2}, and \cref{thm:stability1}) fail to identify the stability, i.e., the corresponding LMI conditions are infeasible.
\end{theorem}

\begin{example}
Here, we will cosider the two examples used in~\cref{ex:7} again.
Let us consider the linear system $\dot x(t) = -x(t)$, which is globally exponentially stable.
For its fuzzy model~\eqref{fuzzy-system} with $A_1 = -1, A_2 = 0, \alpha_1(x) = 1, \alpha_2(x) = 0$, the FLF-based conditions (\cref{lemma:Tanaka2003} and \cref{lemma:Mozelli2009}) are infeasible as shown in~\cref{ex:7}. In this case, $\alpha_1(0) = 1, \alpha_2(0) = 0$, and $A_ 0 = -1$, which is Hurwitz. Therefore, with sufficiently small bounds $b_i , i\in {\cal I}_r$ and $\eta$, \cref{thm:stability1}, \cref{thm:stability2}, and \cref{thm:stability3} are feasible.

If we consider the different MFs ${\alpha _1}(x) = {x^2},{\alpha _1}(x) = 1 - {x^2}$ in~\cref{ex:6}, then ${\alpha _1}(0) = 0,{\alpha _1}(x) = 1$ and $A_0 = 0$. In this case, \cref{thm:stability1}, \cref{thm:stability2}, and \cref{thm:stability3} are feasible with any values of $b_i , i\in {\cal I}_r$ and $\eta$.

Let us consider the nonlinear system~\eqref{nonlinear-system} with $f(x) = \left[ {\begin{array}{*{20}{c}}
{{x_2}}\\
{ - {x_1} + \mu (1 + {x_2}){x_2}/2}
\end{array}} \right]$ and $\mu = -2$, which is locally exponentially stable as shonw in~\cref{ex:7}.
The corresponding fuzzy system~\eqref{fuzzy-system} is defined with~\eqref{eq:14} and ${\alpha _1}(x) = (1 + {x_2})/2,{\alpha _2}(x) = 1 - {\alpha _1}(x)$, ${\cal X} = \{ x \in {\mathbb R}^2:{x_1} \in [ - 1,1],{x_2} \in [ - 1,1]\} $.
Using~\cref{thm:fundamental-limit4}, we can prove that the FLF-based conditions (\cref{lemma:Tanaka2003}, \cref{lemma:Mozelli2009}, and other conditions in the literature) are infeasible.
On the other hand, we have ${\alpha _1}(0) = 1/2,{\alpha _2}(0) = 1/2$, and ${A_0} = \left[ {\begin{array}{*{20}{c}}
0&1\\
{ - 1}&{\mu /2}
\end{array}} \right]$, which is Hurwitz. Therefore, \cref{thm:stability1}, \cref{thm:stability2}, and \cref{thm:stability3} are feasible with some values of $b_i \in [0,1], i\in {\cal I}_r$ and $\eta >0$.

Now, let us consider the MFs ${\alpha _1}(x) = 1 - x_1^2,{\alpha _2}(x) = 1 - {\alpha _1}(x)$ with ${\cal X} = \{ x \in {\mathbb R}^2:{x_1} \in [ - 1,1],{x_2} \in [ - 1,1]\}$. The corresponding nonlinear system~\eqref{nonlinear-system} with $f(x) = \left[ {\begin{array}{*{20}{c}}
{{x_2}}\\
{\mu (1 - x_1^2){x_2} - {x_1}}
\end{array}} \right]$ is the well-known van der Pol oscillator, which is known to be locally asymptotically but non-exponentially stable around the origin.
In this case, \cref{thm:stability1}, \cref{thm:stability2}, and \cref{thm:stability3} are infeasible with any values of $b_i\in [0,1] , i\in {\cal I}_r$ and $\eta>0$.
\end{example}

Before closing this section, let us summarize the advantages of the proposed approaches: 1) As demonstrated in~\cref{thm:converse1} and~\cref{thm:converse2}, the proposed local quadratic stability methods offer necessary and sufficient LMI conditions for ensuring local exponential stability. Furthermore, these conditions are less conservative compared to those based on FLF methods. 2) Moreover, the proposed quadratic approaches can be synergistically combined with FLF methods. The compounded effects of both methods serve to further reduce the conservatism inherent in each individual approach.

\section{Compounding the two effects}\label{sec:combined-approach}
In the preceding section, we established that the proposed QLF approaches are less conservative than the FLF approaches. However, it is important to note that these feasibility inclusions are contingent upon the hyperparameters $\varphi_i,i\in {\cal I}_r$, $b_i,i\in {\cal I}_r$ and $\eta$ being variables that can be optimized. When these hyperparameters are predetermined, the feasibility inclusions no longer hold. In such instances, both the FLF and QLF methods can reduce their inherent conservatism by being compounded with each other. To elaborate, let us consider the following candidate Lyapunov function:
\begin{align}
V(x) = {x^T}{P_0}x + \sum_{i = 1}^r {{\alpha _i}(x){x^T}{P_i}x}.\label{eq:11}
\end{align}
where $P_0 = P_0^T \in {\mathbb R}^{n\times n}, P_i = P_i^T \in {\mathbb R}^{n\times n},i \in {\cal I}_r$, which combines the FLF and QLF candidates
Its time-derivative along the solution is
\begin{align}
\dot V(x) =& \sum_{i = 1}^r {({\alpha _i}(x) - {\alpha _i}(0)){x^T}({P_0}{A_i} + A_i^T{P_0})x}\nonumber\\
&  + {x^T}(A_0^T{P_0} + {P_0}{A_0})x\nonumber\\
&+ \sum\limits_{j = 1}^r {\sum\limits_{i = 1}^r {{\alpha _j}(x){\alpha _i}(x)} } {x^T}({P_j}{A_i} + A_i^T{P_j})x\nonumber\\
& + \sum\limits_{i = 1}^r {{{\dot \alpha }_i}(x){x^T}{P_i}x}.\label{eq:10}
\end{align}

To address the derivative ${{\dot \alpha }_i}(x)$ and the difference ${\alpha _i}(x) - {\alpha _i}(0)$, the techniques involving polytopic and ball bounds discussed in the previous section can be applied. These techniques can be variously combined to yield a range of LMI stability conditions, each with differing degrees of conservativeness. Due to space constraints, this paper will focus solely on the combination of Lemma~\cref{lemma:Mozelli2009} and Theorem~\cref{thm:stability1}. It should be noted that different combinations could lead to novel conditions, representing potential avenues for future research.
\begin{theorem}\label{thm:stability4}
Suppose that the MFs are continuously differentiable.
Moreover, suppose that there exist constants $b_i \in [0,1], \varphi_i > 0,i\in {\cal I}_r$, symmetric matrices $P_0 = P_0^T \in {\mathbb R}^{n\times n}$, $P_i = P_i^T \in {\mathbb R}^{n\times n},i\in {\cal I}_r$, $M=M^T \in {\mathbb R}^{n\times n}$, $N=N^T \in {\mathbb R}^{n\times n}$, and $W=W^T \in {\mathbb R}^{n\times n}$ such that the following LMIs hold:
\begin{align}
&{P_i} + N \succ 0,\quad i \in {\cal I}_r,\label{eq:11}\\
&\sum\limits_{i = 1}^r {{\delta _i}({P_i} + W)}  + \sum\limits_{i = 1}^r {{\alpha _i}(0){P_i}} + P_0 \succ 0,\label{eq:13}\\
&\sum\limits_{k = 1}^r {{\varphi _k}({P_k} + N)}  + \frac{1}{2}\left( {A_i^T{P_j} + {P_j}{A_i} + A_j^T{P_i} + {P_i}{A_j}} \right)\nonumber\\
& + A_0^TP + P{A_0} + \sum\limits_{k = 1}^r {{\delta _k}(A_k^TP + P{A_k} + M)}  \prec 0,\label{eq:12}\\
&\forall ({\delta _1},{\delta _2}, \ldots ,{\delta _r}) \in \{  - {b_1},{b_1}\}  \times \{  - {b_2},{b_2}\}  \times  \cdots  \times \{  - {b_r},{b_r}\}\nonumber\\
&\forall (i,j) \in {\cal I}_r \times {\cal I}_r,\quad i \le j\nonumber
\end{align}

Then, the fuzzy system~\eqref{fuzzy-system} is locally asymptotically stable.
\end{theorem}

The proof is given in Appendix~\ref{appendix:10}. If the LMI condition of~\cref{thm:stability4} is feasible, then it implies that $V(x)  > 0,\forall x \in {\cal H}(b)\backslash \{ 0\} ,{\nabla _x}V{(x)^T}A(\alpha )x < 0,\forall x \in ({\cal H}(b) \cap \Omega (\varphi ))\backslash \{ 0\}$, and any Lyapunov sublevelset such that ${L_V}(c) \subseteq {\cal H}(b) \cap \Omega (\varphi )$ with some $c>0$ is a DA. In what follows, we prove that~\cref{thm:stability4} is less conservative than~\cref{lemma:Mozelli2009} and~\cref{thm:stability1}.
\begin{theorem}[Feasibility inclusion]\label{thm:non-conservatism4}
Suppose that the LMI condition of~\cref{lemma:Mozelli2009} is feasible for some constants $\varphi_i > 0,i\in {\cal I}_r$.
Then, the LMI condition of~\cref{thm:stability4} is also feasible with the same constants $\varphi_i > 0,i\in {\cal I}_r$ and any $b_i\in [0,1],i\in {\cal I}_r$. The conserve does not hold in general.

On the other hand, suppose that the LMI condition of~\cref{thm:stability1} is feasible for some constnats $b_i\in [0,1],i\in {\cal I}_r$. Then, the LMI condition of~\cref{thm:stability4} is also feasible with the same $b_i\in [0,1],i\in {\cal I}_r$ and any $\varphi_i > 0,i\in {\cal I}_r$. The conserve does not hold in general.
\end{theorem}
The proof is given in Appendix~\ref{appendix:11}.

\begin{example}\label{ex:8}
Let us consider the system in~\cref{ex:2} again.
The stability has been checked using~\cref{lemma:Mozelli2009} and~\cref{thm:stability4} for several values of pairs $(a,b)\in [-10,0]\times [0,200]$. For~\cref{lemma:Mozelli2009}, the derivative bounds are set to be $\varphi_i=0.85, i\in {\cal I}_r$. For~\cref{thm:stability4}, we consider $\varphi_i=0.85, i\in {\cal I}_r$ and  $b_i=0.2, i\in {\cal I}_r$.
\cref{fig:7}(a) depicts a set of feasible points for several values of pairs $(a,b)\in [-10,0]\times [0,200]$.
It reveals that the region of~\cref{thm:stability4} includes the region of~\cref{lemma:Mozelli2009}.
\begin{figure}
\centering
\includegraphics[width=0.3\textwidth]{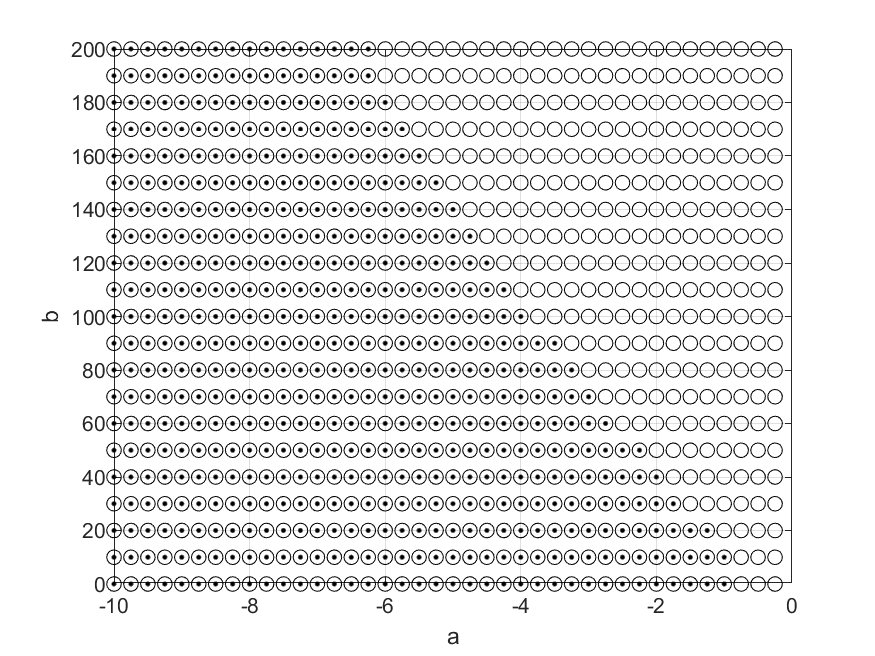}
\caption{\cref{ex:8}: Feasible region for several values of $(a,b)\in [-10,0]\times [0,200]$.
`$\bullet$' indicates the feasible point of~\cref{lemma:Mozelli2009} with $\varphi_i=0.85, i\in {\cal I}_r$.
`$\circ$' indicates the feasible point of~\cref{thm:stability4} with $\varphi_i=0.85, i\in {\cal I}_r$ and $b_i=0.2, i\in {\cal I}_r$.}\label{fig:7}
\end{figure}
\end{example}

\section{Conclusion and future directions}\label{sec:conclusion}
In this paper, we have introduced novel local stability conditions for T-S fuzzy systems based on QLFs. These newly formulated conditions incorporate information from the membership functions (MFs) at the origin. We have rigorously demonstrated that these conditions are both necessary and sufficient for ensuring local exponential stability. Moreover, it has been established that the proposed conditions exhibit less conservatism compared to traditional approaches based on FLFs. Additionally, we have proposed an extended Lyapunov function that combines both quadratic and FLFs, and theoretically proved that this generalization improves the aforementioned methods. To validate the theoretical claims, several illustrative examples have been provided.

The proposed methods hold promise for extension to state-feedback stabilization problems, as well as to the stability and stabilization of discrete-time T-S fuzzy systems. Various avenues exist for further reducing conservatism, including the application of additional LMI techniques to automatically estimate DAs as sublevel sets of Lyapunov functions. These topics constitute promising directions for future research.

\bibliographystyle{IEEEtran}
\bibliography{reference}

\appendices

\section{Proof of~\cref{prop:1}}\label{appendix:2}
Since the derivatives of the MFs, ${\nabla _x}{\alpha_i}(x)$, exist and continuous, and $\cal X$ is compact, ${\nabla _x}{\alpha _i}(x)$ is bounded over $\cal X$, i.e., ${\left\| {{\nabla _x}{\alpha _i}(x)} \right\|_2} \le \gamma ,\forall x \in {\cal X}$ for some $\gamma >0$. Therefore, we have
\begin{align*}
|{\nabla _x}{\alpha_i}(x)^T A(\alpha(x)) x | \le& {\left\| {\nabla _x}{\alpha _i}(x) \right\|_2}{\left\| A(\alpha (x)) \right\|_2}{\left\| x \right\|_2}\\
\le& \gamma \sum_{i = 1}^r {\left\| A_i \right\|_2} {\left\| x \right\|_2},\quad \forall x \in {\cal X},
\end{align*}
where the Cauchy–Schwarz inequality is used in the first inequality, and the triangle inequality is used in the second inequality.
Therefore, we have $\lim_{x \to 0} |{\nabla _x}{\alpha _i}(x)A(\alpha (x))x| = 0,\forall i \in {\cal I}_r$, and one can choose a sufficiently small ball ${\cal B}(\varepsilon)$ with some $\varepsilon >0$ such that
\begin{align*}
|{\nabla _x}{\alpha _i}{(x)^T}A(\alpha (x))x| \le \gamma \sum\limits_{i = 1}^r {{{\left\| {{A_i}} \right\|}_2}} \varepsilon  \le {\varphi _i},\quad \forall x \in {\cal B}(\varepsilon)
\end{align*}
which completes the proof.

\section{Proof of~\cref{thm:conservatism}}\label{appendix:9}
Let us consider the following simple scalar system
\begin{align}
\dot x(t) = f(x(t)) = (2\sin(x(t)) - 1)x(t).\label{eq:nonlinear-system-ex}
\end{align}
Note that the linearization ${\left. {{\nabla _x}f(x)} \right|_{x = 0}} =  - 1$ is Hurwitz.
Therefore, by Lyapunov direct method, the system is locally asymptotically stable around the origin.
A fuzzy system corresponding to~\eqref{eq:nonlinear-system-ex} is given by~\eqref{fuzzy-system} with ${A_1} = 1,{A_2} =  - 1,{\alpha _1}(x) = \sin (x),{\alpha _2}(x) = (1 - \sin (x)), {\cal X} = [0,\pi /2]$.
Now, let us focus on the LMI condition of~\cref{lemma:Mozelli2009}, and suppose that it is feasible, which implies that
\begin{align}
&P_i \succ 0,\quad A_i^T{P_j} + {P_j}{A_i} + A_j^T{P_i} + {P_i}{A_j} \prec 0,\nonumber\\
&\forall (i,j) \in {\cal I}_r \times {\cal I}_r,\quad i \le j.\label{eq:6}
\end{align}
are feasible. Multiplying the first and second inequality by ${\alpha _i}(x){\alpha _j}(x)$ and summing them over all ${\cal I}_r \times {\cal I}_r, i \le j$ lead to
\begin{align}
&P(\alpha (x)) \succ 0,\quad A{(\alpha (x))^T}P(\alpha (x)) + P(\alpha (x))A(\alpha (x)) \prec 0,\nonumber\\
&x \in {\cal X},\label{eq:5}
\end{align}
which means that $A(\alpha (x))$ is Hurwitz for any $x\in {\cal X}$.
As a contraposition, one concludes that if~\eqref{eq:5} does not hold, then~\eqref{eq:6} is infeasible.
Moreover, if~\eqref{eq:6} is infeasible, then so is the LMI condition of~\cref{lemma:Mozelli2009}.
Returning to the fuzzy system example, one can prove that~\eqref{eq:5} does not hold at $x = \pi/2$ because ${\alpha _1}(\pi /2) = \sin (\pi /2) = 1,{\alpha _2}(\pi /2) = 1 - \sin (\pi /2) = 0$, and $A(\alpha (\pi /2)) = 1$. Therefore,~\eqref{eq:6} is infeasible. Therefore, the LMI condition of~\cref{lemma:Mozelli2009} is infeasible as well. This completes the proof.

\section{Proof of~\cref{thm:stability1}}\label{appendix:4}
Suppose that the LMIs in~\cref{thm:stability1} are feasible. Then, it is obvious that~\eqref{eq:thm1:1} holds for all $\delta : = {\left[ {\begin{array}{*{20}{c}}
{{\delta _1}}&{{\delta _2}}& \cdots &{{\delta _r}}
\end{array}} \right]^T} \in H(b)$, which implies that~\eqref{eq:thm1:1} holds for all points in the intersection ${\cal H}(b) \cap \{ \delta  \in {\mathbb R}^r:{\bf{1}}_r^T\delta  = 0\}$, where ${\bf{1}}_r$ is the $r$-dimensional vector with all entries one. Now, setting ${\delta _i} = {\alpha _i}(x) - {\alpha _i}(0)$, assuming $\alpha (x) - \alpha (0) \in {\cal H}(b) \cap \{ \delta  \in {\mathbb R}^r:{\bf{1}}_r^T\delta  = 0\}$, and using the null property~\eqref{eq:null-property}, we have
\begin{align*}
&\sum_{i = 1}^r {({\alpha _i}(x) - {\alpha _i}(0))(A_i^TP + P{A_i})}  + A_0^TP + P{A_0} \prec 0
\end{align*}
for all $x \in {\cal H}(b)$. By multiplying both sides of the above inequality~\eqref{eq:1} by the state $x$ from the right and its transpose from the left and considering~\eqref{eq:3}, one concludes that $\dot V(x(t)) < 0,\forall x(t) \in {\cal H}(b) \backslash \{ 0\}$ holds. This completes the proof.

\section{Proof of~\cref{thm:stability3}}\label{appendix:3}
\begin{lemma}\label{lemma:bounding-lemma}
Given matrices $U,V$ of appropriate dimensions, the following holds for any $G = G^T \succ 0$:
\[
 - U^T G^{-1} U - V^T G V \preceq U^T V + V^T U \preceq U^T G^{-1} U + V^T G V.
\]
\end{lemma}
\begin{proof}
The first inequality comes from $(G^{ - 1/2} U + G^{1/2} V)^T ( G^{ - 1/2} U + G^{1/2} V) =  U^T G^{ - 1} U + U^T V + V^T U + V^T G V \succeq 0$ and the reversed inequality is obtained from $(G^{ - 1/2} U - G^{1/2} V)^T (G^{ - 1/2} U - G^{1/2} V) =  U^T G^{-1} U - U^T V - V^T U + V^T G V\succeq 0$. This completes the proof.
\end{proof}

Let us define
\begin{align*}
{\Gamma ^T}: = \left[ {\begin{array}{*{20}{c}}
{{{(P{A_1} + M)}^T}}& \cdots &{{{(P{A_r} + M)}^T}}
\end{array}} \right],
\end{align*}
and suppose that the LMIs in~\cref{thm:stability3} are feasible.
Applying the Schur complement leads to $\eta G + {\Gamma ^T}{(I \otimes G)^{ - 1}}\Gamma  + P{A_0} + A_0^T P \prec 0$.
Then, for any $x \in {\cal U}(\eta)$, we have
\begin{align*}
&\sum_{i = 1}^r {(\alpha_1(x) - \alpha_1(0))^2} G\\
& + \Gamma^T (I_r \otimes G)^{-1} \Gamma + P A_0 + A_0^T P \prec 0.
\end{align*}
Noting the identity $\sum_{i = 1}^r {{({\alpha _1}(x) - {\alpha _1}(0))}^2} G = {((\alpha (x) - \alpha (0)) \otimes I_n)^T}(I_r \otimes G)((\alpha (x) - \alpha (0)) \otimes {I_n})$ and using~\cref{lemma:bounding-lemma} lead to
\begin{align*}
&((\alpha (x) - \alpha (0)) \otimes {I_n})^T \Gamma  + \Gamma^T((\alpha (x) - \alpha (0)) \otimes {I_n})\\
& + A_0^TP + P{A_0} \prec 0,\quad x \in {\cal U}(\eta ).
\end{align*}
Using the null property~\eqref{eq:null-property} leads to~\eqref{eq:1} for all $x \in {\cal U}(\eta )$.

\section{Proof of~\cref{thm:converse1}}\label{appendix:5}

\begin{lemma}[{Lyapunov's indirect method~\cite[Corollary~4.3]{khalil2002nonlinear}}]
Let us consider the system~\eqref{nonlinear-system}, suppose that $f$ is continuously differentiable in some neighborhood of the origin, and the origin is an equilibrium point of~\eqref{nonlinear-system}.
Then, the origin is locally exponentially stable if and only if the linearization $A = {\left. {{\nabla _x}A(\alpha (x))x} \right|_{x = 0}}$ around the origin is Hurwitz.
\end{lemma}

The sufficiency has been proved in~\cref{thm:stability1}.
For the necessity, suppose that the original nonlinear system is locally exponentially stable around the origin.
Since $A(\alpha(x) )x = f(x),\forall x \in {\cal X}$, one gets
\begin{align*}
{\nabla _x}f(x) =& {\nabla _x}A(\alpha )x\\
=& {\nabla _x}\left( {\sum\limits_{i = 1}^r {{\alpha _i}(x){A_i}x} } \right)\\
=& \sum\limits_{i = 1}^r {{A_i}x {\nabla _x}{\alpha _i}(x)}  + \sum\limits_{i = 1}^r {{\alpha _i}(x){A_i}}
\end{align*}
Therefore, by Lyapunov's indirect method, one concludes that
\begin{align*}
{\left. {{\nabla _x}A(\alpha(x) )x} \right|_{x = 0}} =& {\left. {\sum_{i = 1}^r {{A_i}x {\nabla _x}{\alpha _i}(x)} } \right|_{x = 0}} + {\left. {\sum\limits_{i = 1}^r {{\alpha _i}(x){A_i}} } \right|_{x = 0}}\\
=& \sum\limits_{i = 1}^r {{\alpha _i}(0){A_i}}  = {A_0}
\end{align*}
is Hurwitz stable, i.e., there exists a Lyapunov matrix $P^* \succ0$ such that $A_0^TP^* + P^*{A_0} =  - Q \prec 0$
for any $Q\succ 0$. Letting $M = I$ and $P = P^*$ leads to
\begin{align*}
&\sum_{i = 1}^r {{\delta _i}(A_i^TP + P{A_i} + M)}\\
=& \sum_{i = 1}^r {{\delta _i}(A_i^TP + P{A_i})}  + \sum\limits_{i = 1}^r {{\delta _i}M} \\
\le& e(b)I,
\end{align*}
where $e(b): = \sum_{i = 1}^r {{{\max }_{i \in {\cal I}_r}}{b_i}\xi_i }$ and ${\xi _i}: = \max \{ |{\lambda _{\max }}(A_i^TP + P{A_i})|,|{\lambda _{\max }}( - A_i^TP - P{A_i})|\}  + 1$. Then, one can choose sufficiently small $b_i,i\in {\cal I}_r$ such that $e(b)I \prec Q$, implying $A_0^TP + P{A_0} + e(b)I \prec 0$ so that $A_0^TP + P{A_0} + \sum_{i = 1}^r {{\delta _i}(A_i^TP + P{A_i} + M)}  \prec 0$ for all $({\delta _1},{\delta _2}, \ldots ,{\delta _r}) \in \{  - {b_1},{b_1}\}  \times \{  - {b_2},{b_2}\}  \times  \cdots  \times \{  - {b_r},{b_r}\} $. This completes the proof.

\section{Proof of~\cref{thm:converse2}}\label{appendix:7}

The sufficiency has been proved in~\cref{thm:stability2} and~\cref{thm:stability3}.
For the necessity, suppose that the original nonlinear function is locally asymptotically stable around the origin.
Then, following similar lines as in the proof of~\cref{thm:converse1}, there exists a Lyapunov matrix $P^* \succ0$ such that $A_0^TP^* + P^*{A_0} =  - Q \prec 0$ for any $Q\succ 0$.

{\it Proof of the first statement:} Let $M = -c I$ and choose a sufficiently large $c>0$ such that~\eqref{eq:thm2:1} is satisfied.
Next, since $A_0^TP^* + P^*{A_0} \prec 0$ holds, letting $P = P^*$, one can choose a sufficiently small $b_i \in [0,1],i\in {\cal I}_r$ so that~\eqref{eq:thm2:2} holds. This completes the first part of the proof.

{\it Proof of the second statement:} Let us define
\begin{align*}
{\Gamma ^T}: = \left[ {\begin{array}{*{20}{c}}
{{{(P{A_1} + M)}^T}}& \cdots &{{{(P{A_r} + M)}^T}}
\end{array}} \right].
\end{align*}
Applying the Schur complement to~\eqref{eq:8} leads to
\begin{align}
&\eta G + \Gamma^T {(I \otimes G)^{ - 1}}\Gamma + P{A_0} + A_0^TP \prec 0.\label{eq:9}
\end{align}
Let $P=P^*$, $M = I$, $G = \varepsilon^{-1} I$, and $\eta \leq \varepsilon^2$.
When $\varepsilon \to 0$, the first and second terms in~\eqref{eq:9} vanish. Therefore, one can choose a sufficiently small $\varepsilon >0$ such that the inequality in~\eqref{eq:9} holds. In other words, we can choose a sufficiently small $\eta>0$ so that the LMI condition in~\cref{thm:stability3} is feasible. This completes the proof.

\section{Proof of~\cref{thm:non-conservatism2}}\label{appendix:6}
Suppose that the LMI condition of~\cref{lemma:Tanaka2003} is feasible for some constants $\varphi_i> 0,i\in {\cal I}_r$.
Since $P_i + M \succeq 0, i\in {\cal I}_r$, the LMIs in~\eqref{eq:4} imply that
\begin{align*}
A_i^T{P_j} + {P_j}{A_i} + A_j^T{P_i} + {P_i}{A_j} \prec 0,\quad \forall (i,j) \in {\cal I} \times {\cal I}_r,i \le j.
\end{align*}

Multiplying $\alpha_i(x)\alpha_j(x)$ on both sides of the above inequalities and summing them over all $i,j \in {\cal I}_r, i \le j$ lead to
\begin{align*}
A{(\alpha (x))^T}P(\alpha (x)) + P(\alpha (x))A(\alpha (x)) \prec 0,\forall x \in {\cal X}
\end{align*}
with $P(\alpha (x)) = \sum_{i = 1}^r {{\alpha _i}(x){P_i}}$. Letting $x=0$ yields $A{(\alpha (0))^T}P(\alpha (0)) + P(\alpha (0))A(\alpha (0)) \prec 0 \Rightarrow A_0^TP + P{A_0} \prec 0$
with $P = P(\alpha (0))$, which means that $A_0$ is Hurwitz. Now, it is clear that one can choose sufficiently small $b_i > 0,i\in {\cal I}_r$ so that~\eqref{eq:thm1:1} holds.

To prove the converse, let us consider the fuzzy system~\eqref{fuzzy-system} with ${A_1} = 1, {A_2} =  - 1,{\alpha _1}(x) = \sin (x),{\alpha _2}(x) = (1 - \sin (x)),{\cal X} = [0,\pi /2]$.
Noting ${\alpha _1}(0) = 0,{\alpha _2}(0) = 1$, we have ${A_0} = {\alpha _1}(0){A_1} + {\alpha _2}(0){A_2} =  - 1$, which is Hurwitz. Therefore, following the same steps as in the proof of~\cref{thm:converse1}, one concludes that the LMI condition of~\cref{thm:stability1} is also feasible for sufficiently small $b_i,i\in {\cal I}_r$.
For this system, we can prove that the LMI condition of~\cref{lemma:Tanaka2003} is not feasible for any constants $\varphi_i > 0,i\in {\cal I}_r$. In particular, ${\varphi _1}{P_1} + {\varphi _2}{P_2} + A_1^T{P_1} + {P_1}{A_1} + A_1^T{P_1} + {P_1}{A_1} \prec 0$ implies ${P_1} <  - \frac{{{\varphi _2}}}{{1 + {\varphi _1}}}{P_2}$, which cannot be satisfies with $P_1 >0, P_2 >0$. This completes the proof.

\section{Proof of~\cref{thm:fundamental-limit2}}\label{appendix:8}
Let us consider a nonlinear system~\eqref{nonlinear-system} which is locally asymptotically but non-exponentially stable.
By contradiction, suppose that the FLF approaches (\cref{lemma:Tanaka2003} and \cref{lemma:Mozelli2009}) and the QLF approaches (\cref{thm:stability1}, \cref{thm:stability2}, and \cref{thm:stability3}) found a feasible solution.
This implies that there exists some Lyapunov sublevel set ${L_V}(c)$ with some $c>0$ such that it is a DA.
To proceed further, let us consider the fuzzy Lypaunov function $V(x) = {x^T}P(\alpha (x))x$.
Then, for any $x(t)\in {L_V}(c)$, we have $\frac{d}{{dt}}V(x(t)) \le  - \varepsilon x{(t)^T}x(t) < 0$ for all $x(t) \ne 0$ and for a sufficiently small constant $\varepsilon>0$. Noting ${\min _{i \in {\cal I}_r}}{\lambda _{\min }}({P_i}){x^T}x \le {x^T}P(\alpha (x))x \le {\max _{i \in {\cal I}_r}}{\lambda _{\max }}({P_i}){x^T}x$, this implies that
\begin{align*}
\frac{d}{{dt}}V(x(t)) \le  - \left( {\frac{\varepsilon }{{{{\max }_{i \in {\cal I}_r}}{\lambda _{\max }}({P_i})}}} \right)V(x(t)),\forall x(t) \ne 0
\end{align*}
and equivalently, $V(x(t)) \le \exp \left( { - \frac{{\varepsilon t}}{{{{\max }_{i \in {\cal I}_r}}{\lambda _{\max }}({P_i})}}} \right)V(0)$. Then, it follows that ${\left\| {x(t)} \right\|_2} \le \sqrt {\frac{1}{{{{\min }_{i \in {\cal I}_r}}{\lambda _{\min }}({P_i})}}\exp \left( { - \frac{{\varepsilon t}}{{{{\max }_{i \in {\cal I}_r}}{\lambda _{\max }}({P_i})}}} \right)V(0)}$.
Therefore, \eqref{nonlinear-system} is locally exponentially stable, and it contradicts with the hypothesis that \eqref{nonlinear-system} is locally asymptotically but non-exponentially stable.
This completes the proof.

\section{Proof of~\cref{thm:stability4}}\label{appendix:10}
The proof is a combination of the proof of~\cref{lemma:Mozelli2009} and~\cref{thm:stability1}.
Suppose that the LMIs in~\cref{thm:stability4} are feasible.
Then, one can let ${\varphi _i} = {{\dot \alpha }_i}(x(t)),i \in {\cal I}_r$ and ${\delta _i} = {\alpha _i}(x(t)) - {\alpha _i}(0),i \in {\cal I}_r$ in~\eqref{eq:12}. Using the null effects $\sum_{k = 1}^r {{{\dot \alpha }_k}(x(t))}  = 0$, $\sum_{i = 1}^r {({\alpha _i}(x(t)) - {\alpha _i}(0))}  = 0$, using the bound~\eqref{eq:11}, multiplying the inequality~\eqref{eq:12} by $\alpha_i(x(t))\alpha_j(x(t))$, and summing it over all $(i,i)\in {\cal I}_r \times {\cal I}_r, i \leq j$ lead to
\begin{align*}
&\sum_{i = 1}^r {({\alpha _i}(x(t)) - {\alpha _i}(0))({P_0}{A_i} + A_i^T{P_0})}  + A_0^T{P_0} + {P_0}{A_0}\\
& + \sum_{j = 1}^r {\sum_{i = 1}^r {{\alpha _j}(x(t)){\alpha _i}(x(t))} } ({P_j}{A_i} + A_i^T{P_j})\\
& + \sum_{i = 1}^r {{{\dot \alpha }_i}(x(t)){P_i}}  \prec 0,\quad \forall x \in {\cal H}(b) \cap \Omega (\varphi )
\end{align*}
which holds if and only if the time-derivative of the Lyapunov function along the solution~\eqref{eq:10} is negative definite.

On the other hand, letting ${\delta _i} = {\alpha _i}(x(t)) - {\alpha _i}(0),i \in {\cal I}_r$ in~\eqref{eq:13} and using the null effect $\sum_{i = 1}^r {({\alpha _i}(x(t)) - {\alpha _i}(0))}  = 0$, one has ${P_0} + \sum_{i = 1}^r {{\alpha _i}(x){P_i}}  \succ 0,\forall x \in {\cal H}(b) \Leftrightarrow V(x) > 0,\forall x \in {\cal H}(b)\backslash \{ 0\}$. This completes the proof.

\section{Proof of~\cref{thm:non-conservatism4}}\label{appendix:11}
Suppose that the LMI condition of~\cref{lemma:Mozelli2009} is feasible for some constants $\varphi_i > 0,i\in {\cal I}_r$.
Then, for any $b_i \in [0,1],i\in {\cal I}_r$, it can be easily proved that with $P_0= \varepsilon_1 I$, $M= \varepsilon_2 I$, and sufficiently small $\varepsilon_1, \varepsilon_2 >0$,~\eqref{eq:12} holds. Moreover, it is clear that~\eqref{eq:11} and~\eqref{eq:13} hold.

On the other hand, suppose that the LMI condition of~\cref{thm:stability1} is feasible for some constnats $b_i\in [0,1],i\in {\cal I}_r$. Then, for any $\varphi_i \in [0,1],i\in {\cal I}_r$, it can be easily proved that with $P_i= \varepsilon_i I,i\in {\cal I}_r$, $N= \varepsilon_{r+1} I$ and sufficiently small $\varepsilon_i>0,i\in {\cal I}_r$, $\varepsilon_{r+1}>0$,~\eqref{eq:12} holds. Moreover, it is clear that~\eqref{eq:11} and~\eqref{eq:13} hold.
This completes the proof.

\end{document}